\documentclass[12pt,a4paper]{article}
\usepackage[margin=2cm]{geometry}
\usepackage{authblk}

\usepackage[breaklinks]{hyperref}
\usepackage{booktabs}
\usepackage{txfonts}

\usepackage{todonotes,soul}
\usepackage{braket}

\usepackage{amsthm}
\usepackage{amsmath}
\usepackage{amssymb}
\usepackage{amsfonts}
\usepackage{mathbbol}
\usepackage{caption}
\usepackage{bm}
\usepackage{xcolor}
\usepackage{mathrsfs}

\newcommand{\dV}{d\mathrm{Vol}}

\newcommand{\be}{\begin{equation}}
	\newcommand{\ee}{\end{equation}}
\newcommand{\bea}{\begin{eqnarray}}
	\newcommand{\eea}{\end{eqnarray}}
\def\bml{\begin{subequations}}
	\def\blea{\bml\begin{eqnarray}}
		\def\eml{\end{subequations}}
	\def\elea{\end{eqnarray}\eml}

\newtheorem{theorem}{Theorem}[section]

\newtheorem{lemma}[theorem]{Lemma}
\newtheorem{proposition}[theorem]{Proposition}
\newtheorem{definition}[theorem]{Definition}

\def\fmax{\phi_{\text{max}}}

\newcommand{\nord}[1]{{:}#1{:}}
\newcommand{\coin}[1]{\left[\!\left[#1\right]\!\right]}

\DeclareMathOperator{\supp}{supp}
\newcommand{\pFq}[5]{{}_{#1}F_{#2}\left(\begin{array}{c} #3 \\ #4 \end{array}; #5\right)}

\newcommand{\KdF}[9]{F^{#1}_{#2}\left(\begin{array}{c} #3: #4 \mathbin{;} #5 \\ #6: #7\mathbin{;} #8 \end{array}; #9\right)}

\begin{document}
	\title{A semiclassical singularity theorem}
	\author{Christopher J. Fewster}
	\affil{Department of Mathematics, University of York, Heslington, York YO10 5DD, United Kingdom}
	\author{Eleni-Alexandra Kontou}
	\affil{Department of Mathematics, University of York, Heslington, York YO10 5DD, United Kingdom}
	\affil{Department of Physics, College of the Holy Cross, Worcester, Massachusetts 01610, USA}
	\affil{ ITFA and GRAPPA, Universiteit van Amsterdam,
		Science Park 904, Amsterdam, the Netherlands}
	
	\maketitle

	\begin{abstract}
		Quantum fields do not satisfy the pointwise energy conditions that are assumed in the original singularity theorems of Penrose and Hawking. Accordingly, semiclassical quantum gravity lies outside their scope. Although a number of singularity theorems have been derived under weakened energy conditions, none is directly derived from quantum field theory. Here, we employ a quantum energy inequality satisfied by the quantized minimally coupled linear scalar field to derive a singularity theorem valid in semiclassical gravity. By considering a toy cosmological model, we show that our result predicts timelike geodesic incompleteness on plausible timescales with reasonable conditions at a spacelike Cauchy surface. 
	\end{abstract}

\section{Introduction}
\label{sec:introduction}

Singularity theorems represent an effort to answer a central question in gravitational physics: under which conditions do cosmological or astrophysical systems originate or end in a singularity? In effect: are singularities inevitable in our universe? The famous   singularity theorems of Penrose~\cite{Penrose:1964wq} and Hawking~\cite{Hawking:1966sx} address this question in general relativity by taking geodesic incompleteness as the defining characteristic of a singularity, and showing that this occurs, from suitable starting situations, for any matter model obeying energy conditions broadly expressing the idea of local energy positivity.
In particular, Pen{\-}rose's theorem on null geodesic incompleteness was based on the
null energy condition (NEC), that the stress-energy tensor should obey $T_{\mu \nu} \ell^\mu \ell^\nu \geq 0$ for all null vectors $\ell^\mu$, while Hawking's result on timelike geodesic incompleteness required the strong energy condition (SEC), that
\be
\label{eqn:SEC}
\rho_U := T_{\mu \nu} U^\mu U^\nu-\frac{T^{\mu}_{\phantom{\mu}\mu}}{n-2}  \geq 0  \qquad \text{for all unit timelike vectors $U$,}
\ee
where $n$ is the spacetime dimension. Here, $\rho_U$ is the \emph{effective energy density} (EED) seen by an observer with velocity $U^\mu$ (cf.\ Ref.~\cite{Pirani:1956tn}, which refers to
the ``effective density of gravitational mass''). Via the Einstein equation, the SEC is equivalent to the geometrical condition $R_{\mu\nu}U^\mu U^\nu \leq 0$ on the Ricci tensor (our conventions are listed at the end of the introduction).

It has long been recognized that the energy conditions are too restrictive. 
Even the classical massive minimally coupled scalar field violates the SEC and the non-minimally coupled scalar field violates the NEC, while quantum fields violate all pointwise energy conditions \cite{Epstein:1965zza}. Reviews concern{\-}ing the status of energy conditions include Refs.~\cite{Curiel:2014zba, Kontou:2020bta}, while  Ref.~\cite{Senovilla:2018aav} is a comprehensive review of all aspects of singularity theorems. Accordingly, several authors have focused on relaxing the energy condition in singularity theorems looking instead at weaker, averaged energy conditions \cite{Tipler:1978zz,Borde:1987qr, Roman:1988vv,Wald:1991xn}. Most of that work focuses on Penrose-type (null geodesic incompleteness) results and the half-averaged null energy condition. More recently, Ref.~\cite{Fewster:2010gm} presented proofs of both Hawking-type (timelike geodesic incompleteness) and Penrose-type singularity theorems with weakened energy conditions. Based on their results, Ref.~\cite{Brown:2018hym} proved a Hawking-type singularity theorem for the classical Einstein-Klein-Gordon field. The essential technical ingredient in this work, as with almost all of the previous literature concerning weakened energy conditions, is the analysis of a 
Riccati differential inequality derived from the Raychaudhuri equations~\cite{Ehlers:1993gf}.
However, it has been realised very recently~\cite{Fewster:2019bjg} that \emph{index form} methods provide a much simpler route to obtaining singularity theorems with weakened energy hypotheses, with the advantage that they also provide quantitative estimates of the timescale on which the singularity occurs. Index form methods appear in some textbook accounts of the original singularity theorems~\cite{ONeill,Kriele} and are also employed in generalizations of the 
singularity theorems in other directions -- see e.g.,~\cite{GallowaySenovilla:2010,AazamiJavaloyes}. An interesting account of the
interplay between Riccati and index-form methods in a slightly different context can be found in Ref.~\cite{EhrlichKim:2009}.
 
But what about quantum fields? Reconciliation of the energy conditions and quantum field theory began with the seminal works of Ford~\cite{Ford:1978qya,Ford:1990id}. Motivated by an apparent conflict between negative energy densities and the second law of thermodynamics, he
suggested and later proved that quantum field theory contains mechanisms that impose
fundamental restrictions on the possible magnitude and duration of any negative energy densities or fluxes due to the quantum field; these restrictions -- \emph{quantum energy inequalities} (QEIs) -- have been proved to hold for a range of different quantum field theories in flat and curved spacetimes. See~\cite{Fewster2017QEIs,Kontou:2020bta} for recent reviews.

An important question is whether QEIs can be used as hypotheses for singularity theorems and this has been addressed to some extent. Both Refs.~\cite{Fewster:2010gm} and \cite{Fewster:2019bjg} introduced singularity theorems with hypotheses inspired by QEIs. For example, Ref.~\cite{Fewster:2019bjg} proved a Hawking-type singularity theorem in spacetime dimension $n=2m$ ($m\in\mathbb{N}$) under the following geometrical condition on the Ricci tensor: for any unit-speed timelike geodesic $\gamma:I\to M$, where $I$ is a compact interval in $\mathbb{R}$, it is assumed that there are non-negative constants 
$Q_m(\gamma)$ and $Q_0(\gamma)$ of appropriate dimensions so that  
the inequality
\be
\label{eqn:Riccigen}
\int_I  f(t)^2 R_{\mu\nu} \dot{\gamma}^\mu \dot{\gamma}^\nu|_{\gamma(t)} \,dt\le  Q_m(\gamma) \|f^{(m)}\|^2 + Q_0(\gamma) \|f\|^2\,,
\ee
holds for all smooth real-valued $f$ supported in the interior of $I$. Here $\|\cdot\|$ denotes the standard norm of $L^2(I)$. Under this condition, and additional assumptions on the extrinsic curvature of a spacelike Cauchy surface, it was shown that timelike goedesic completeness must fail. Quantitative estimates of the required initial contraction and timescales for the singularity were obtained in various cases, 
including models drawn from cosmology.

However, none of the conditions used in the works just mentioned were derived directly from quantum field theories. Indeed, QEIs on the quantised EED -- needed for Hawking-type results -- were only derived recently~\cite{Fewster:2018pey}. What we will do in this paper is to use those
QEIs as the basis for a singularity theorem valid in semiclassical quantum gravity. There remains a slight gap. The bounds obtained in~\cite{Fewster:2018pey} hold in flat and curved spacetime backgrounds but in the latter case require a reference state to be chosen.
(In Minkowski spacetime we simply adopt the Poincar\'e invariant vacuum state as the reference.) As we will argue, however, the contribution from the reference state can be neglected on sufficiently small scales, on which the bound is dominated by terms with a universal form (essentially equal to the bound in Minkowski spacetime). The upshot is that we 
assume the bound~\eqref{eqn:Riccigen} for solutions to semiclassical Einstein--Klein--Gordon system, with
constants $Q_m$ and $Q_0$ obtained from the QEI and other assumptions, but restricting to functions $f$ whose support extends over sufficiently short intervals of proper time. This timescale is expected to shrink as a singularity is approached. Ideally, it would be derived from the QEI itself -- a line of work that
will be pursued elsewhere -- but in the present paper it must be supplied as an additional piece of information. Thus we will speak of a geodesic being $T$-regular, where $T:I\to \mathbb{R}$, if the inequality~\eqref{eqn:Riccigen} holds for functions $f$ supported in intervals of width $T_t$ centered at $t$, allowing $t$ to vary within $I$ (see Definition~\ref{def:Treg} for a precise statement). 

One approach (which had been our initial plan) is only to assume \eqref{eqn:Riccigen} for $f$ with support duration below a fixed timescale determined by the largest curvature scales encountered along the geodesic $\gamma$. This corresponds to $T$-regularity for a constant function $T$. In principle a direct application of the results in \cite{Fewster:2019bjg} would now give a singularity theorem predicting geodesic incompleteness with a side assumption that the curvature remain below some threshold chosen sufficiently high that spacetimes violating the assumption might be judged to be `singular for practical purposes'. However we found this to require such large extrinsic curvatures at the initial surface as to be of limited practical use. What we are going to do here is to allow the `sufficiently small' timescale to shrink linearly in proper time along the geodesic $\gamma:[0,\tau]\to \mathcal{M}$, by taking 
\begin{equation}\label{eq:Ttintro}
	T_t=T_0(1-t/\tau)
\end{equation} 
for some $T_0>0$. On physical grounds, spacetimes for which this is not valid should exhibit very large curvatures along future complete geodesics and therefore also exhibit singular behaviour `for practical purposes'. 
We then employ a partition of unity argument to control the Ricci tensor over the whole of $\gamma$, where the bump function in the partition near $t\in [0,\tau)$ has support width less than $T_t$. 

Simplifying slightly, 
our result is essentially the following. Consider any solution to the semiclassical Einstein--Klein--Gordon equations in which the expected magnitude of the scalar field does not exceed some threshold $\phi_{\text{max}}$. It is assumed that the spacetime of the solution is globally hyperbolic and that the state of the quantum field is Hadamard. Let $S$ be a spacelike Cauchy surface in the spacetime and suppose that, for some $T_0, \tau>0$, all future-directed unit-speed timelike geodesics of length $\tau$ emanating normally from $S$ are $T$-regular where $T_t$ is given by~\eqref{eq:Ttintro}. Suppose also that the SEC holds for a period of proper time $\tau_0$ at the start of each such geodesic. If the extrinsic curvature of $S$ is sufficiently large and negative, where the threshold is determined by $T_0$, $\tau$, $\tau_0$ and the spacetime dimension $n$, then no future-directed timelike curve emanating from $S$ can have proper duration more than $\tau$ and the spacetime
is future-timelike geodesically incomplete. Of course, a sufficiently large extrinsic curvature $-K>(n-1)/\tau_0$ at the initial surface will result, by Hawking's original theorem, in timelike geodesic incompleteness within the timescale $\tau_0$ on which the SEC is assumed. It is important to clarify that
there are indeed regimes where our threshold for $-K$ is much lower than $(n-1)/\tau_0$.
Another way of interpreting our result is that if the initial extrinsic curvature exceeds our threshold
in a solution to the semiclassical Einstein--Klein--Gordon theory obeying the SEC near $S$, then one or more of the following must occur: (a) the spacetime is timelike geodesically incomplete; (b) the field strength must somewhere exceed $\phi_{\text{max}}$; (c) the timescale on which~\eqref{eqn:Riccigen} holds is somewhere smaller than $T_t=T_0(1-t/\tau)$ at proper time $t$ along a geodesic emanating normally from $S$. Depending on the choice of parameters, possibilities (b) and (c) may be regarded as indicating the onset of early-universe conditions and therefore of singular behaviour for practical purposes. 

The paper is structured as follows. Sec.~\ref{sec:quantization} briefly summarises the quantization procedure used in \cite{Fewster:2018pey} to obtain a QEI on the EED (called a Quantum Strong Energy Inequality or QSEI).
The QSEI is recalled in Sec.~\ref{sec:SQEI}, where we also explain why the bounds of the type~\eqref{eqn:Riccigen} may be expected on sufficiently small scales, leading to the idea of a $T$-regular geodesic. In Sec.~\ref{sec:part} we define the partition of unity used to discuss averages over long timescales, deferring many details to the Appendix. Sec.~\ref{sec:sing} presents the singularity theorem that constitutes our main result. It is then shown in Sec.~\ref{sec:cosmology} how quantitative information can be extracted concerning the time at which geodesic incompleteness occurs and the required extrinsic curvature at the initial surface. In the context of a toy cosmological model we show that plausible extrinsic curvatures lead to singular behaviour on timescales that exceed the lifetime of our own universe by as little as one order of magnitude. As a definition of `plausible' parameters, we use values that
are predicted to have occurred in our universe according to $\Lambda$CDM cosmology fitted with PLANCK data~\cite{aghanim2018planck}. We emphasise that the purpose of these calculations is to show that results of the type we have developed do indeed apply beyond the domain of applicability of Hawking's result, and also as a proof of principle that the quantitative information produced by results of the type we have developed might have practical use.
A summary and discussion of our results in Sec.~\ref{sec:conclusions} concludes the paper.

We work in $n$ spacetime dimensions unless otherwise stated and employ $[-,-,-]$ conventions in the Misner, Thorne and Wheeler classification \cite{MTW}. That is, the metric signature is $(+,-,-, \dots)$, the Riemann tensor is defined as $R^{\phantom{\lambda \eta\nu}\mu}_{\lambda \eta\nu}v^\nu=(\nabla_\lambda \nabla_\eta-\nabla_\eta \nabla_\lambda)v^\mu$, and the Einstein equation is $G_{\mu \nu}=-8\pi T_{\mu \nu}$. The d'Alembertian is written $\Box_g = g^{\mu\nu}\nabla_\mu\nabla_\nu$.  For the most part, we adopt units in which $G_N=c=1$.

\section{Quantization}
\label{sec:quantization}

Throughout this work, we assume that the spacetime is a smooth $n$-dimensional Lorentzian manifold $(\mathcal{M},g)$ that is globally hyperbolic. We consider the minimally coupled classical scalar field $\phi$ with Lagrangian density 
\be
L[\phi]=\frac{\sqrt{-g}}{2} \left( (\nabla \phi)^2-M^2 \phi^2 \right) \,,
\ee
where $M$ has dimensions of inverse length, equal to the reduced Compton length of massive particles in the corresponding quantum field theory. The field equation and stress-energy tensor are 
\be
\label{eqn:KleinGordon}
(\Box_g+M^2)\phi=0 
\ee
and
\be
\label{eqn:tmunu}
T_{\mu \nu}=(\nabla_\mu \phi)(\nabla_\nu \phi)+\frac{1}{2} g_{\mu \nu} (M^2 \phi^2-(\nabla \phi)^2) \,,
\ee
while the EED defined by Eq.~\eqref{eqn:SEC} becomes
\be
\label{eqn:classEED}
\rho_U=U^\mu U^\nu (\nabla_\nu \phi) (\nabla_\mu \phi) -\frac{1}{n-2} M^2 \phi^2 \,.
\ee

Quantization follows exactly the procedure described in Ref.~\cite{Fewster:2018pey} so we will be brief and direct the reader there for more details. It is based on  the algebraic approach (see \cite{KhavkineMoretti-aqft} for a review) and starts by introducing a unital $*$-algebra $\mathscr{A}(\mathcal{M})$ on our manifold $\mathcal{M}$, so that self-adjoint elements of $\mathscr{A}(\mathcal{M})$ are observables of the theory. The algebra is generated by elements $\Phi(f)$ where $f \in C_0^\infty(\mathcal{M})$, which represent smeared quantum fields and obey the following relations
\begin{itemize}
	\item
	\textbf{Linearity} \\
	The map $f\rightarrow\Phi(f)$ is complex-linear,
	\item
	\textbf{Hermiticity} \\
	$\Phi(f)^* = \Phi(\overline{f})$ \qquad $\forall f \in C^{\infty}_0(\mathcal{M})$,
	\item
	\textbf{Field Equation} \\
	$\Phi((\Box_g+M^2) f) = 0$  \qquad $\forall f \in C^{\infty}_0(\mathcal{M})$,
	\item
	\textbf{Canonical Commutation Relations} \\
	$\left[\Phi(f), \Phi(h)\right] = iE (f,h)\mathbb{1}$  \qquad $\forall f,h \in C^{\infty}_0(\mathcal{M})$.
\end{itemize} 
Here, $C_0^\infty(\mathcal{M})$ denotes the space of smooth functions of compact support from $\mathcal{M}$ to $\mathbb{C}$ and $E(x,y)=E^{A}(x,y)-E^R(x,y)$ is an antisymmetric bi-distribution equal to the difference of the advanced and retarded Green functions of $\Box_g+M^2$, which exist due to global hyperbolicity of the spacetime.  

In the algebraic approach, states of the theory are described by linear maps $\omega:\mathscr{A}(\mathcal{M})\to\mathbb{C}$ that are normalized ($\omega(\mathbb{1})=1$) and positive ($\omega(A^*A)\ge 0$ for all $A\in\mathscr{A}(\mathcal{M})$). The expression $\omega(A)$ is interpreted as the expectation value of $A\in\mathscr{A}(\mathcal{M})$ in state $\omega$. Physical conditions must be specified to narrow down the choice of states, and a particularly well-behaved class of states is given by those that are Hadamard 
(see \cite{KhavkineMoretti-aqft}, and also~\cite{moretti2021global} for a recent improvement to the original definition~\cite{KayWald:1991}).  

The algebra $\mathscr{A}(\mathcal{M})$ only contains polynomials built from elements of the form $\Phi(f)$. For example, $\mathscr{A}(\mathcal{M})$ contains elements of the form
\begin{equation}\label{eq:Wkex}
	\nord{\Phi^{\otimes 2}}_\omega(f\otimes f) = \Phi(f)\Phi(f) - W(f,f)\mathbb{1}
\end{equation}
for any test function $f$ and any state $\omega$, where 
\begin{equation}
	W(f,h) = \omega(\Phi(f)\Phi(h))
\end{equation} 
is the two-point function of $\omega$. These elements have the property that $\langle\nord{\Phi^{\otimes 2}}_\omega(f\otimes f)\rangle_\omega=0$. 

To form a Wick square, one needs to replace $f\otimes f\in C_0^\infty(\mathcal{M}\times \mathcal{M})$ by a suitable compactly supported distribution supported on the diagonal (and also address finite renormalisation freedoms). This idea can be implemented in the case where $\omega$ is quasifree\footnote{A state is quasifree if all its odd $n$-point functions vanish and all its even $n$-point functions can be expanded as sums of products of the two-point function according to Wick’s theorem.} and Hadamard and leads to the definition of an extended algebra $\mathscr{W}(\mathcal{M})$, described in~\cite{Hollands:2001nf}, that contains $\mathscr{A}(\mathcal{M})$ as a subalgebra, but also contains elements such as Wick polynomials and the smeared stress-energy tensor.
For example, the quadratic Wick polynomials in the field and its derivatives needed to define the quantized stress-energy tensor are given as follows.

Let $f^{\mu_1\cdots \mu_r \nu_1\cdots \nu_s}$ be a smooth compactly supported tensor field and define a compactly
supported distribution $T^{r,s}[f]$ on $\mathcal{M}\times \mathcal{M}$ by
\begin{equation}\label{eq:Trs_def}
	T^{(r,s)}[f](S) = \int_{\mathcal{M}} \dV\, f^{\mu_1\cdots \mu_r \nu_1\cdots \nu_s} \coin{(\nabla^{(r)}\otimes \nabla^{(s )}) S^{\mathrm{sym}}}_{\mu_1\cdots \mu_r \nu_1\cdots \nu_s}.
\end{equation} 
Here,
$S^{\mathrm{sym}}(x,y)=\frac{1}{2} (S(x,y)+S(y,x))$ is the symmetric part of $S\in C^\infty(\mathcal{M}\times \mathcal{M})$, while $\nabla^{(r)}$ is a symmetrised $r$-th order covariant derivative and the double square brackets $\coin{\cdot}$ in the integrand denote a coincidence limit. Then we obtain a smeared Wick polynomial
\begin{equation}
	\nord{\nabla^{(r)}\Phi \nabla^{(s)}\Phi}_\omega (f):= \nord{\Phi^{\otimes 2}}_\omega(T^{r,s}[f]) \,,
\end{equation}
which has a vanishing expectation value in the reference state $\omega$. 
However, as there is no preferred choice of state in a general curved spacetime (\cite{Fewster:2011pe}; see \cite{Fewster_artofstate:2018} for discussion) it is desirable to seek prescriptions leading to local and covariant Wick polynomials. This can be done in various ways and the renormalisation freedom cannot be removed completely. Here we follow the axioms of Hollands and Wald~\cite{Hollands:2001nf,Hollands:2004yh}. In this scheme, Wick polynomials obey a Leibniz rule for differentiation of products, but the field equation cannot be used freely. 

The axioms are enough to define the stress-energy tensor $T_{\mu\nu}$ for perturbatively interacting theories, up to finite renormalisation freedoms, which might be fixed by reference to the expectation values obtained in certain standard states on standard spacetimes, or other criteria. In the case of the linear real scalar field, the finite renormalisation freedom reduces to linear combinations of the tensors $M^4 g_{\mu\nu}$,  $M^2 G_{\mu\nu}$, $I_{\mu\nu}$ and $J_{\mu\nu}$ where the latter are
conserved symmetric curvature tensors arising by functional derivatives of $R^2$ and $R_{\mu\nu}R^{\mu\nu}$ with respect to the metric. 
It is natural to use some of this freedom to choose a prescription in which the Minkowski vacuum state has vanishing expectation value, for example.  At any rate, we assume that a prescription has been fixed as part of the physical specification of the matter model. 
Once this is done, the semiclassical Einstein equation (SEE) for the Einstein--Klein--Gordon system with vanishing cosmological constant can be formulated as follows. 
 \begin{definition}
 	\label{def:solSEE}
 	Fix a locally covariant prescription for the stress-energy tensor of the massive minimally coupled scalar field. A \emph{solution to the SEE} is a triple  
 	$(\mathcal{M},g,\omega)$ where $(\mathcal{M},g)$ is a globally hyperbolic spacetime of even dimension $n=2m$ and $\omega$ is a Hadamard state of the quantised minimally coupled Klein--Gordon field on $(\mathcal{M},g)$, such that
 	\be
 	\label{eqn:SEE}
 	G_{\mu \nu}=  - 8\pi \langle  T_{\mu \nu}  \rangle_\omega   
 	\ee
 	 holds everywhere on $\mathcal{M}$.
 \end{definition}
The system~\eqref{eqn:SEE} is a highly complex set of equations linking the evolution of the metric to that of the two-point function of the quantum state $\omega$. If a prescription has been chosen in which the Minkowski vacuum state has vanishing expectated stress-energy tensor, then Minkowski space furnishes a rather trivial example solution. More generally, the analysis of these equations is highly nontrivial. 
Recently Sanders \cite{Sanders:2020osl} has succeeded in classifying all static symmetric semiclassical solutions to the Einstein--Klein--Gordon system with spatial topology $\mathbb{S}^3$ and arbitrary cosmological constant and allowing for nonminimal coupling. Before that, the existence of solutions has been shown in situations with high levels of symmetry, for example in flat FLRW spacetimes (see \cite{Pinamonti:2013wya} and \cite{Gottschalk:2018kqt}). In this paper, however, we will formulate a singularity theorem for
solutions to the SEE without further discussing the question of existence.

Clearly, if $(\mathcal{M},g,\omega)$ is a solution to the SEE, then the Ricci tensor obeys
\begin{equation}\label{eq:RicciEED}
	R_{\mu\nu} U^\mu U^\nu = - 8\pi \langle \rho_U\rangle_\omega
\end{equation}
for any timelike unit vector $U^\mu$, where $\rho_U$ is the renormalized EED.
As a smeared field, $\rho_U$ takes the form
\be
\label{eqn:EED1}
\rho_U(f) = (\nabla_\mu\Phi\nabla_\nu\Phi)(U^\mu U^\nu f)-\Phi^2\left( \frac{1}{n-2} M^2  f\right) \,,
\ee
where $U$ is any unit timelike vector field and $f$ is any smooth compactly supported function on $\mathcal{M}$.\footnote{At nonminimal coupling there are terms such as $\Phi (\Box_g+\xi R+M^2)\Phi$. These do not necessarily vanish in the Hollands--Wald scheme, but they turn out to be state-independent and therefore cancel in differences of expectation values in Hadamard states.}   

For further manipulations, the following point-splitting formula is useful:
\begin{align}
	\langle (\nabla^{(r)}\Phi\nabla^{(s)}\Phi)(f)\rangle_{\omega} - 
	\langle (\nabla^{(r)}\Phi\nabla^{(s)}\Phi)(f)\rangle_{\omega_0} &= 
	\langle \nord{(\nabla^{(r)}\Phi\nabla^{(s)}\Phi)}_{\omega_0}(f)\rangle_{\omega} \notag\\
	&= T^{r,s}[f](W-W_0) \,,
\end{align}
where both $\omega$ and $\omega_0$ are quasifree Hadamard states
with respective two-point functions $W$ and $W_0$. Because $\omega$ and $\omega_0$ are Hadamard, the difference $\nord{W}=W-W_0$ is necessarily smooth, so the right-hand side is well-defined. Here, and for the rest of this paper, we denote all quantities normal-ordered relative to $\omega_0$ by $\nord{X}$, rather than $\nord{X}_{\omega_0}$. It follows that
\begin{equation}\label{eq:rhosplit}
	\langle \rho_U(f)\rangle_\omega = \langle \rho_U(f)\rangle_{\omega_0} + \langle \nord{\rho_U}(f)\rangle_\omega, 
\end{equation}
where
\begin{equation}\label{eqn:rhoquantum}
	\langle \nord{\rho_{U}}(f)\rangle_\omega = \int_{\mathcal{M}}\dV\, f \left[\coin{ (\nabla_{U} \otimes \nabla_{U} ) \nord{W} } - \left(\frac{M^2}{n-2} \right) \coin{\nord{W}}\right]\,.
\end{equation}
The quantity appearing in square brackets will be denoted 
$\langle \nord{\rho_{U}}\rangle_\omega$ by a slight abuse of notation.

It is important to note that Eq.~\eqref{eq:RicciEED} assumes a particular choice to fix finite renormalisation freedoms has been made for $\rho_U$, inherited from the choices made for the stress-energy tensor. 
In equation~\eqref{eq:rhosplit} the term $\nord{\rho_U}$ is independent of such choices because they are manifested in terms that are multiples of the identity operator.

	\section{Quantum strong energy inequality}
\label{sec:SQEI}

We can now turn to the QEIs satisfied by the quantized EED derived in
Ref.~\cite{Fewster:2018pey}, which established a variety of lower bounds
on averages of $\langle \nord{\rho_U}\rangle_\omega$ along worldlines or over spacetime volumes, allowing for possibly nonminimal coupling to the scalar curvature. Here, we will employ the worldline result for the minimally coupled massive field. 

Let $\gamma$ be a smooth geodesic parametrized by proper time $\tau$
and let $U^\mu$ be a timelike unit vector field agreeing with $\dot{\gamma}^\mu$ on $\gamma$. For any real valued test function   $f \in C_0^\infty (\mathbb{R}, \mathbb{R})$, we write
\be 
\langle \nord{\rho_U} \circ \gamma \rangle_\omega (f^2)=\int d\tau f^2 (\tau) \langle  \nord{\rho_U}\rangle_\omega   (\gamma(\tau)) \,.
\ee 
It was shown in Theorem~5 of~\cite{Fewster:2018pey} that the QEI
\be
\label{eqn:clinemin}
\langle \nord{\rho_U} \circ \gamma \rangle_\omega (f^2) \geq - \left[ \int_0^\infty \frac{d\alpha}{\pi} \phi^*((\nabla_{U} \otimes \nabla_{U}) \, W_0)(\overline{f_\alpha},f_\alpha)+\frac{M^2}{n-2}  \langle \nord{\Phi^2}  \circ \gamma \rangle_\omega(f^2) \right] 
\ee  
holds for all Hadamard states $\omega$, where $f_\alpha(\tau)=e^{i\alpha\tau}f(\tau)$ and $\phi:\mathbb{R}\times\mathbb{R}\to\mathcal{M}\times\mathcal{M}$ by $\phi(\tau,\tau')=(\gamma(\tau),\gamma(\tau'))$.
Consequently the Hadamard-renormalized EED obeys the inequality
\begin{align}
	\label{eqn:clineminHad}
	\langle \rho_U \circ \gamma \rangle_\omega (f^2)  +\frac{M^2}{n-2} \langle \Phi^2  \circ \gamma \rangle_\omega(f^2) & \geq 
	\langle \rho_U \circ \gamma \rangle_{\omega_0} (f^2)+ \frac{M^2}{n-2} \langle \Phi^2  \circ \gamma \rangle_{\omega_0}(f^2) \notag\\
	&\qquad\qquad
	-   \int_0^\infty \frac{d\alpha}{\pi} \phi^*((\nabla_{U} \otimes \nabla_{U}) \, W_0)(\overline{f_\alpha},f_\alpha) 
\end{align} 
for all Hadamard states $\omega$, for any choice of Hadamard reference $\omega_0$ with two-point function $W_0$, and for all real-valued test functions $f$.

This bound is not easy to use in practice, due to the lack of 
closed-form expressions for Hadamard two-point functions that could define a reference state. An alternative would be to derive an `absolute' QSEI that avoids the need for a reference state, and this will be pursued elsewhere.
Here, we will instead work around the obstacle by making a physically motivated  approximation. Namely, if the support of $f$ is sufficiently small, 
the integration in the last term of~\eqref{eqn:clineminHad} involves the two-point function $W_0(x,y)$ near to coincidence along $\gamma$. In this regime, the two-point function is dominated by its leading singularity, which is universal across all Hadamard states and is shared, in particular, by the Minkowski vacuum two-point function of a massless scalar field. For test functions of sufficiently small support -- where `sufficiently small' is determined by the reference state $\omega_0$ and the spacetime geometry but \emph{not} the states $\omega$ of interest -- the bound~\eqref{eqn:clineminHad} continues to hold for all $\omega$ if the integral on the right-hand side is replaced by
\be 
\int_0^\infty \frac{d\alpha}{\pi} \int_{\mathbb{R}^2} dt\,dt' 
\partial_t\partial_{t'}W_{\text{Mink}}(t,\boldsymbol{0};t',\boldsymbol{0})
\overline{f_\alpha(t)} f_\alpha(t') 
\ee
where 
\be
W_{\text{Mink}}(x,x')=\hbar \int \frac{d^{n-1}\boldsymbol{k}}{(2\pi)^{n-1}} \frac{e^{-ik\cdot(x-x')}}{2k}
\ee
is the massless Minkowski vacuum $2$-point function and $k^\mu=(\|\boldsymbol{k}\|,\boldsymbol{k})$.  Planck's constant appears explicitly as it has not been set to unity.
The integral may be evaluated as
\begin{align}
\label{eqn:minkder}
\hbar \int_0^\infty  \frac{d\alpha}{\pi}  \int \frac{d^{n-1}\boldsymbol{k}}{(2\pi)^{n-1}} 
\int_{\mathbb{R}^2} dt\,dt' \frac{k}{2} e^{-i(\alpha+k)(t-t')} f(t)f(t') 
&= \hbar \frac{S_{n-2}}{(2\pi)^n} \int_0^\infty d\alpha\int_0^\infty dk\,
k^{n-1}|\hat{f}(\alpha+k)|^2 \nonumber\\
&=  \hbar \frac{S_{n-2}}{(2\pi)^n} \int_0^\infty du \int_0^u dk\,
k^{n-1}|\hat{f}(u)|^2 \nonumber\\
& =\hbar \frac{S_{n-2}}{(2\pi)^n} \int_0^\infty du\, \frac{u^n}{n}  |\hat{f}|^2(u)  \,,
\end{align}
where $S_{n-2}$ is the volume of the $(n-2)$-dimensional standard unit sphere
and $\hat{f}$ is the Fourier transform $\hat{f}(u)=\int dt\, e^{-iut}f(t)$. 
In the first step, we have used the definition of the transform and integrated out the angular dependence, and in the second we have changed variables from $\alpha$ and $k$ to $u=\alpha+k$ and $k$. 
Restricting to spacetimes with even dimension of at least $4$, writing $n=2m$ for $m \in \mathbb{N}$ and $m\ge 2$, we may use  the fact that $ \hat{f'}(u)=-i u \hat{f}(u)$ to obtain
\be
\label{eqn:extint}
\hbar \frac{S_{n-2}}{(2\pi)^n} \int_0^\infty du\, \frac{u^n}{n}  |\hat{f}|^2(u) =\hbar \frac{S_{2m-2}}{4m(2\pi)^{2m}} \int_{-\infty}^\infty du |\widehat{f^{(m)}}|^2 (u)=\hbar \frac{\pi S_{2m-2}}{2m(2\pi)^{2m}} \int_{-\infty}^\infty dt |f^{(m)}|^2  \,,
\ee
where we have also extended the integral to the full line, using the fact that $|\hat{f}|^2$ is real and even, and used Parseval's theorem in the final step.
As we have argued, this expression can be used -- at least for sufficiently small support width -- as a replacement for the integral in~\eqref{eqn:clineminHad}. We can say a bit more about the magnitude of this term in relation to the others, because smooth functions $f$, compactly supported in an interval $I$ of length $\tau$, obey the Rayleigh--Ritz inequality 
\be
\int_{-\infty}^\infty dt |f^{(m)}|^2  \ge \frac{\lambda^{(m)}_{\min}}{\tau^{2m}}\int_{-\infty}^\infty dt |f|^2\,,
\ee
where $\lambda^{(m)}_{\min}$ is the minimum eigenvalue of the operator $(-1)^m d^{2m}/dx^{2m}$ on $(0,1)$, subject to Dirichlet boundary conditions. (See~\cite{Fewster:2006kt} for other applications of this observation.)
By contrast, the first two terms in~\eqref{eqn:clineminHad} can be bounded below by 
\be
\inf_{I}\left( \langle \rho_U \circ \gamma \rangle_{\omega_0} + \frac{M^2}{n-2} \langle \Phi^2  \circ \gamma \rangle_{\omega_0}\right)\int_{-\infty}^\infty dt |f|^2\,,
\ee
so it is clear that the third term dominates when $\tau$ is sufficiently small. In other words, the first two terms can be made less in magnitude than any prescribed multiple of the third, for small enough $\tau$. Thus for any constant $C>1$, and any $t\in\mathbb{R}$, there is a
timescale $T_t>0$ so that an inequality 
\begin{align}
	\label{eqn:clineminHad3}
	\langle \rho_U \circ \gamma \rangle_\omega (f^2)  +\frac{M^2}{n-2} \langle \Phi^2  \circ \gamma \rangle_\omega(f^2) & \geq -C \hbar \frac{\pi S_{2m-2}}{2m(2\pi)^{2m}} \int_{-\infty}^\infty dt |f^{(m)}|^2  \,,
\end{align} 
holds for all Hadamard states $\omega$ and all real-valued $f$ supported in the interval $(t-T_t/2,t+T_t/2)$. Thus \eqref{eqn:clineminHad3} holds for all real-valued $f$ belonging to
\begin{equation}
	\mathcal{D}_T:= \{f\in C_0^\infty(\mathbb{R}): \text{$\supp f \subset (t-T_t/2,t+T_t/2)$ for some $t\in\mathbb{R}$} \}\,.
\end{equation}
As, in fact, the QEIs derived in Ref.~\cite{Fewster:2018pey} are not expected to be optimally sharp, we may without loss replace $C$ by $1$. The timescale $T_t$ depends on the curve $\gamma$, the spacetime geometry, the mass parameter $M$, and the reference state $\omega_0$. The effect of the mass parameter may be significant because QEI bounds for massive fields are exponentially smaller than those of massless fields -- see~\cite{EvesonFewster:2007} for detailed analysis. Therefore our use of the massless bound is quite conservative and may be expected to hold over reasonable length scales. Furthermore, as we are free to choose any Hadamard state $\omega_0$ as the reference we could in principle choose a different reference state for each value of $t\in\mathbb{R}$ so as to maximise $T_t$. 

Of course it is far from trivial to compute $T_t$ in this way and no  sufficient conditions on $T_t$ are known for~\eqref{eqn:clineminHad3} to hold (with $C=1$) for all real-valued $f\in \mathcal{D}_T$. Such conditions would presumably involve the injectivity radius near $\gamma(t)$ and derivatives of curvature tensors in a neighbourhood thereof. Extrapolating from~\cite{Kontou:2014tha} it may be that only the first three derivatives of the Riemann tensor are relevant at leading order. Lacking sufficient conditions of this type, we turn the problem around and make the following definition.

\begin{definition}
	\label{def:Treg}
	For $\tau \in [0,\infty]$ let $\gamma:[0,\tau]\cap [0,\infty)\to \mathcal{M}$ be a unit-speed timelike geodesic. 
	Let $T:t\mapsto T_t$ be a strictly positive function on $[0,\tau)$ with dimensions of time. 
	We say that $\gamma$ is \emph{$T$-regular} if
	(a) $\gamma$ may be extended to a geodesic $[-\tau_*,\tau]\cap \mathbb{R} \to \mathcal{M}$  where 
	$\tau_*=\sup_{t\in[0,\tau]} (\tfrac{1}{2} T_t-t)$ and (b) 
	the EED of the real scalar field with mass $M$  obeys
	\be\label{eq:EEDfinal}
	\langle \rho_U \circ \gamma \rangle_\omega (f^2)  +\frac{M^2}{n-2} \langle \Phi^2  \circ \gamma \rangle_\omega(f^2) \geq -  \hbar \frac{\pi S_{2m-2}}{2m(2\pi)^{2m}} \int_{-\infty}^\infty dt |f^{(m)}|^2  
	\ee
	for all real-valued $f\in \mathcal{D}_T\cap C_0^\infty(-\tau_*,\tau)$ and all Hadamard states $\omega$.  
\end{definition}

The purpose of condition (a) is to ensure that every interval $(t-T_t/2,t+T_t/2)$ lies in the domain of the geodesic, in order that condition (b) be well-defined (clearly $\tau_*>0$ because $\tau_0>0$). Note that if $\gamma$ is $T$-regular, then it is also $\hat{T}$-regular for any strictly positive
function $t\mapsto \hat{T}_t \le T_t$.

The above definition will be important in the singularity theorems for semiclassical quantum gravity to be proved in Sec.~\ref{sec:sing}, where we consider $T$-regular geodesics $\gamma:[0,\tau]\to \mathcal{M}$ determined by the function $T_t=(1-t/\tau)T_0$. By Def.~\ref{def:Treg}, any solution to the SEE satisfies the geometric condition
\be
\label{eqn:Rmnmin}
\int dt f^2(t) R_{\mu \nu} \dot{\gamma}^\mu \dot{\gamma}^\nu  \leq   \hbar \frac{ S_{2m-2}}{m(2\pi)^{2m-2}} \int_{-\infty}^\infty dt |f^{(m)}|^2 +\frac{4\pi M^2}{m-1} \langle \Phi^2  \circ \gamma \rangle_\omega(f^2) 
\ee
along any $T$-regular unit-speed timelike geodesic, for all $f\in \mathcal{D}_T\cap C_0^\infty(-\tau_*,\tau)$. Note that the timescale $T_t$ on which~\eqref{eq:EEDfinal} holds becomes vanishingly small as $t\to\tau$, which means that the constraint it provides becomes progressively weaker as $\gamma(\tau)$ is approached. 

The second term of the LHS of Eq.~\eqref{eqn:Rmnmin} is state dependent; however, we can restrict $\omega$ to a class of Hadamard states for which the field's magnitude has a finite maximum magnitude. 

\begin{definition}
	\label{def:fsolSEE}
	A $\fmax$-solution to the SEE is a solution to the SEE according to Def.~\ref{def:solSEE} where the state $\omega$ is additionally restricted to a class of Hadamard states for which
	\be
	\label{eqn:maxf}
	\left| \langle \colon \Phi^2 \colon \circ \gamma\rangle_{\omega} \right| \leq \phi_{\max}^2 \,,
	\ee
	where $\fmax$ is a finite constant.
\end{definition}

The above assumptions allow us to make the following statement, which is immediate from Eq.~\eqref{eqn:Rmnmin} using Def.~\ref{def:fsolSEE}.
\begin{proposition}
	\label{the:QEIfin}
	For any $\fmax$-solution to the SEE according to Def.~\ref{def:fsolSEE} on a $T$-regular geodesic according to Def.~\ref{def:Treg} the inequality 
	\be
	\label{eqn:singboundmin}
	\int dt f^2(t) R_{\mu \nu} \dot{\gamma}^\mu \dot{\gamma}^\nu \leq  \frac{ \hbar S_{2m-2}}{m(2\pi)^{2m-2}}  || f^{(m)} ||^2+\frac{4\pi M^2\phi_{\max}^2}{m-1} ||f||^2 
	\ee
	holds for all real-valued $f\in \mathcal{D}_T\cap C_0^\infty(-\tau_*,\tau)$,
	where the norms are those of $L^2(0,\infty)$.\footnote{We only use the inequality for real-valued $f$, but it extends immediately to complex-valued $f$ on considering the real and imaginary parts separately and replacing $f^2$ by $|f(t)|^2$ on the LHS.}
\end{proposition}

	\section{Partition of unity}
\label{sec:part}

Let $\tau\in (0,\infty]$ and suppose $\gamma:[0,\tau]\cap [0,\infty) \to \mathcal{M}$ is a $T$-regular unit-speed timelike geodesic. According to Proposition~\ref{the:QEIfin} a $\fmax$-solution to the SEE satisfies the QEI of Eq.~\eqref{eqn:singboundmin} on $\gamma$. By definition, this means that the QEI holds for test functions of sufficiently small support, where a neighborhood of $t\in [0,\tau)$ is deemed sufficiently small if it is contained in $(t-\tfrac{1}{2}T_t,t+\tfrac{1}{2}T_t)\cap [-\tau_*,\tau)$. However, the singularity theorems require control over averages of the EED taken over the whole of $[0,\tau)$.
These averages may be addressed by choosing a partition of unity to break the long-term average into pieces each of which is sufficiently small that~\eqref{eqn:singboundmin} holds.

Accordingly, we seek smooth compactly supported functions $\phi_n\in C_0^\infty(-\tau_*,\tau)$ such that 
\be\label{eq:partition}
\sum_{n=1}^\infty \phi_n^2=1 
\ee
holds on $[0,\tau)$ and with $\supp \phi_n\subset [t_{n-1},t_n]$,
where the strictly increasing sequence $t_n$ ($n\in\mathbb{N}_0$) is chosen so that $t_n\to\tau$ as $n\to\infty$ and
\be\label{eq:intervalnests}
[t_{n-1},t_n] \subset [t_n-T_{t_n}/2,t_n+T_{t_n}/2]
\ee
for each $n$, which implies $\phi_n\in \mathcal{D}_T\cap C_0^\infty(-\tau_*,\tau)$.

For each $f\in C_0^\infty(0,\tau)$, we may make a decomposition $f^2=\sum_{n=1}^\infty (f\phi_n)^2$ 
and apply Proposition~\ref{the:QEIfin} to each function $f\phi_n$ separately. Summing up, we obtain the inequality 
\be\label{eq:EEDpou}
\int_0^\tau dt f^2(t) R_{\mu \nu} \dot{\gamma}^\mu \dot{\gamma}^\nu \leq Q_0  ||f||^2 + Q_m \sum_{n=1}^\infty  \|(f\phi_n)^{(m)}\|^2  \,,
\ee 
for any $f\in C_0^\infty(0,\tau)$, where the sum contains only finitely many nonzero terms and 
\be
\label{eqn:QgQf}
Q_m= \frac{\hbar S_{2m-2}}{(2\pi)^{2m-2}}   \,, \qquad \text{and} \qquad Q_0=\frac{4\pi M^2\phi_{\max}^2}{m-1}  
\ee
(recall that $m\ge 1$ so there is no conflict in the above definition).
For suitably chosen $\phi_n$, the second term in~\eqref{eq:EEDpou} may be estimated in terms of a weighted Sobolev norm.

As the construction of the functions $\phi_n$ and the derivation of the required estimates is quite complex, we summarise the main points here and provide the details in Appendix~\ref{app:det}. First, in Section~\ref{ssec:POUconstruction}, we will construct functions $\phi_n$ obeying the conditions above, for any positive, nonincreasing continuous function $T_t$. Here we choose $t_n$ to satisfy the equation $t_{n+1}-t_n=\tfrac{1}{2}T_{t_{n+1}}$ for $n\ge 0$ and $t_0=-\tfrac{1}{2}T_0$, from which $t_1=0$ follows; it is shown in particular that $t_n\to\tau$ as $n\to\infty$. The corresponding $\phi_n$ are defined by 
\be
\label{eqn:phi}
\phi_n(t)=\begin{cases} 
	F((t-t_n)/T_{t_n}+\tfrac{1}{2}) \,, & t< t_n \,,\\
	G((t-t_n)/T_{t_{n+1}})  & t \geq t_n  \,,
\end{cases}
\ee 
where $F$ and $G$ are the functions $F(x)=\sin(\theta(x))$ and $G(x)=\cos(\theta(x))$, and $\theta\in C^\infty(\mathbb{R})$ is the smooth  nonnegative function given by 
\be
\label{eqn:theta}
\theta(x)=A\int_0^x  H(x') H(1/2-x')e^{-1/x'} e^{-1/(1/2-x')}dx' \,,
\ee
where $H(x)$ is the Heaviside step function. Here, $A\approx 3.2392 \times 10^4$ is chosen so that $\theta(1/2)=\pi/2$. The function $\theta$ interpolates smoothly from $0$ to $\pi/2$, as shown in Fig.~\ref{fig:thetaplot}. 

\begin{figure}
	\centering
	\includegraphics[height=8cm]{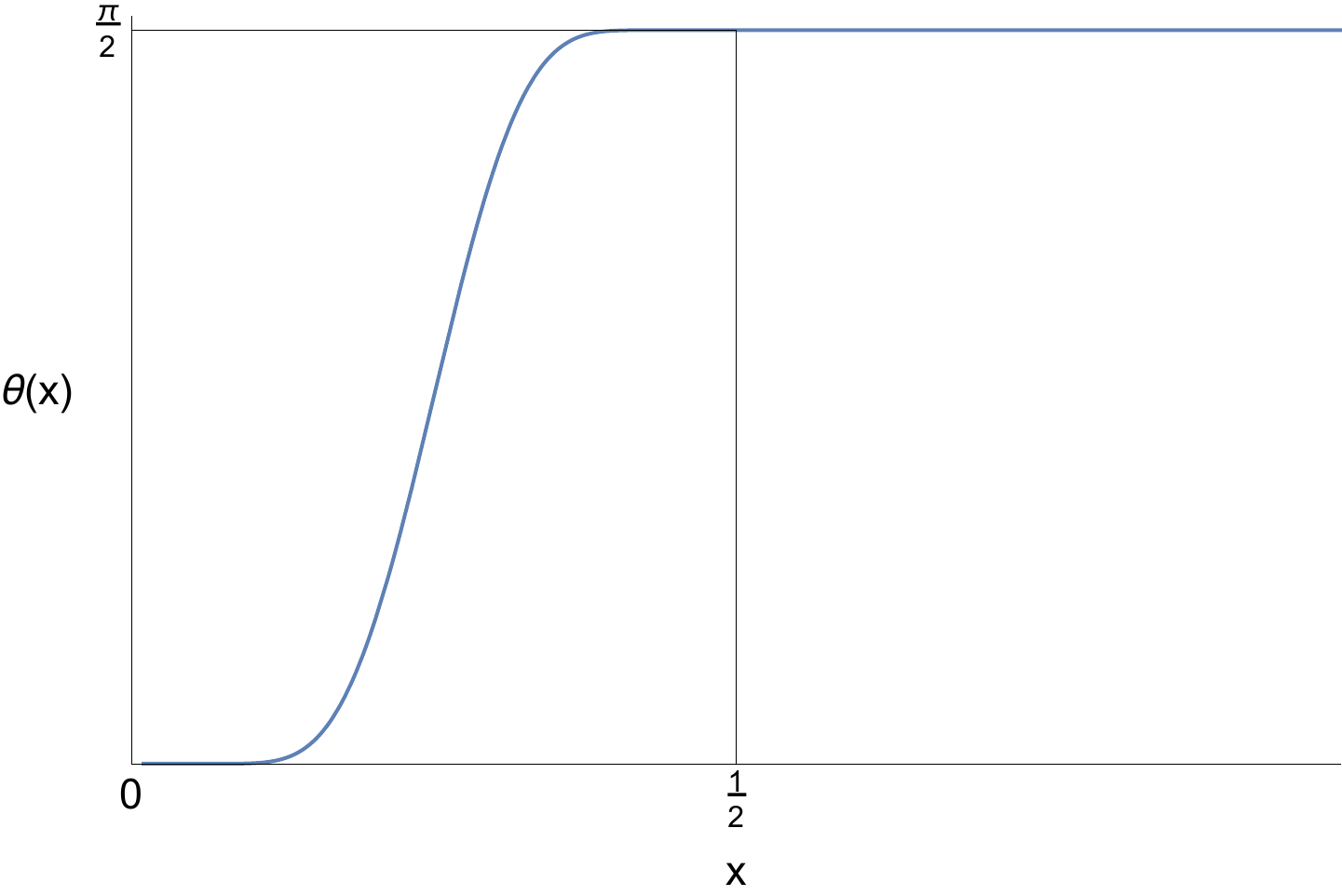}
	\caption{The function $\theta(x)$ defined in Eq.~\eqref{eqn:theta}.}
	\label{fig:thetaplot}
\end{figure}

In Section~\ref{ssec:exampletn} we will consider examples of $T_t$ for both finite and infinite intervals. It will also be shown that other examples can be obtained by reverse engineering the function $T_t$ from a sequence $t_n$. A function of interest, appropriate for a finite interval $[0,\tau]$, is $t\mapsto T_t= T_0(1-t/\tau)$. For that choice we have
\be
\label{eqn:tn}
t_n=\tau \left(1-\left(\frac{2\tau}{T_0+2\tau}\right)^{n-1}\right) \,,
\ee
for all $n\in \mathbb{N}_0$. The functions $\phi_n$ have vanishing support width in the limit $t\to \tau$ as illustrated in Fig.~\ref{fig:phin}.

\begin{figure}
	\centering
	\includegraphics[height=8cm]{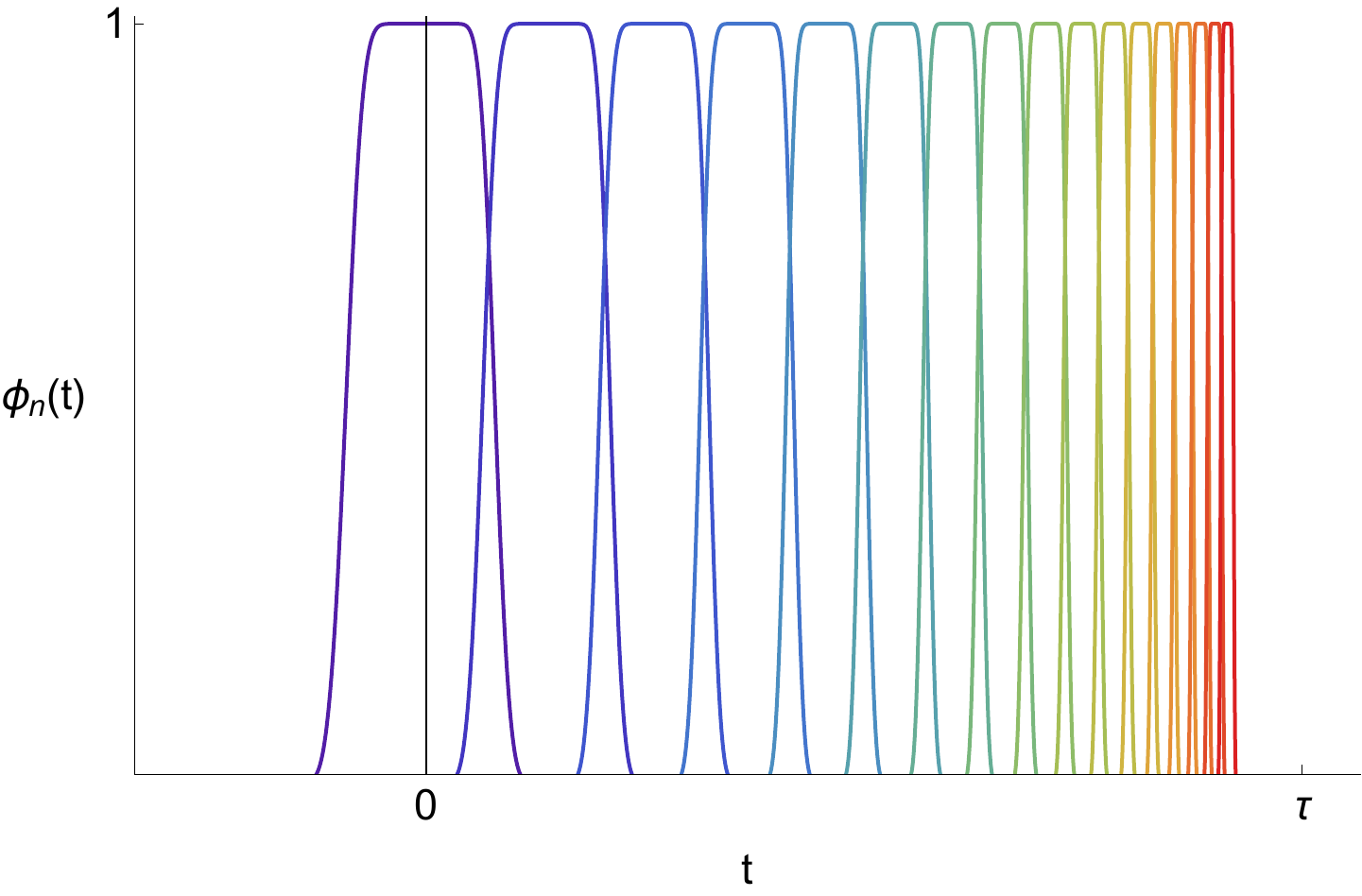}
	\caption{The series of functions $\phi_n$ for the linear $T_t=T_0(1-t/\tau)$. The center of each function is given by $t_n$ (Eq.~\eqref{eqn:tn}) and their support shrinks to a point as $t \to \tau$. }
	\label{fig:phin}
\end{figure}

In Section~\ref{ssec:Sobolev} we estimate the second term of Eq.~\eqref{eq:EEDpou}. Our aim is to estimate this quantity in terms of (weighted) Sobolev norms
of $f$, keeping reasonably explicit control of any constants appearing.
We will do that both for a partition of unity corresponding to a general $T$-regular geodesic and a specific choice of $T_t$ for a finite interval. In particular, for the choice of $T_t= T_0(1-t/\tau)$ and choosing  $\tau_0$ with $0<\tau_0<t_2=\tau T_0/(T_0+2\tau)$, we will show that
\begin{equation} 
	\label{eqn:fphiestim}
	\sum_{n=1}^\infty \|(f \phi_n)^{(m)}\|^2 \leq 2   \left[\sum_{j=0}^m 
	H_j \binom{m}{j} \left( \frac{\|\chi f^{(m-j)}\|}{(2t_2)^j} + c^j\left\|\frac{(1-\chi)f^{(m-j)}}{(\tau-\cdot)^j}\right\|\right)
	\right]^2 \equiv |||f|||^2 \,,
\end{equation}
where $\chi$ is the characteristic function of $[0,\tau_0]$, 
\be
\label{eqn:Hmax}
H_i = \max \{\| F^{(i)}\|_\infty, \|G^{(i)}\|_\infty\} \,,
\ee
and
\be
c=\frac{2\tau^2}{(2\tau-T_0)T_0} \,.
\ee
The upshot of this analysis is expressed by the following theorem.
	\begin{theorem}
		\label{the:Riccfin}
	Let $0<\tau<\infty$, $T_0>0$ and $0<\tau_0<\tau T_0/(T_0+2\tau)$, define $T_t= T_0(1-t/\tau)$.
	For any $\fmax$-solution to the SEE according to Def.~\ref{def:fsolSEE} on a $T$-regular geodesic according to Def.~\ref{def:Treg}, the inequality
	\be
	\label{eqn:Riccfin}
	\int d\tau f^2(t) R_{\mu \nu} \dot{\gamma}^\mu \dot{\gamma}^\nu \leq Q_m |||f|||^2+Q_0 ||f||^2   
	\ee
	holds for all real-valued $f\in W_0^m(0,\tau)$,
	where $|||f|||^2$ is given in Eq.~\eqref{eqn:fphiestim} and the norms are those of $L^2(0,\infty)$.
	
\end{theorem}
Here, we recall that the Sobolev space $W_0^m(0,\tau)$ is the completion of the space $C_0^\infty(0,\tau)$ in the norm $\|f\|_m = (\|f\|^2 + \tau^{2m}\|f^{(m)}\|^2)^{1/2}$.
In particular, its elements are functions with $m-1$ absolutely continuous\footnote{The absolutely continuous functions on $[0,\tau]$ those arising as integrals of $L^1$ (and hence in particular $L^2$) functions.}
derivatives which all
vanish (including the function itself) at the endpoints, and whose $m$'th distributional derivative may be identified with a element of $L^2(0,\tau)$.
\begin{proof}
	For $f\in C_0^\infty(0,\tau)$, the inequality~\eqref{eqn:Riccfin} follows immediately from Proposition~\ref{the:QEIfin} using Eqs.~\eqref{eq:EEDpou} and \eqref{eqn:fphiestim}.
	Now let $f\in W_0^m(0,\tau)$ and let $f_n\in C_0^\infty(0,\tau)$ be a sequence converging to
	$f$ in $W_0^m(0,\tau)$. We may apply~\eqref{eqn:Riccfin} to each $f_n$, and observe that 
	$\int d\tau f_n(\tau)^2 R_{\mu\nu}\dot{\gamma}^\mu\dot{\gamma}^\nu\to \int d\tau f(\tau)^2 R_{\mu\nu}\dot{\gamma}^\mu\dot{\gamma}^\nu$ and $\|f_n\|\to \|f\|$ as $n\to\infty$. Therefore
	\eqref{eqn:Riccfin} will apply to $f$ provided one also has $|||f_n|||\to |||f|||<\infty$. 
	As $\|\chi f_n^{(m-j)}\|\to\|\chi f^{(m-j)}\|$ for all $0\le j\le m$, attention can be focused on the remaining terms contributing to $|||\cdot|||$. These are controlled using the higher-order Hardy inequality (proved in Appendix~\ref{ssec:Hardy})
	\begin{equation}\label{eq:higherHardy}
		\int_{\tau_0}^\tau \frac{|h(t)|^2}{(\tau-t)^{2j}} \, dt \le \frac{4^j}{((2j-1)!!)^2}\int_{\tau_0}^\tau |h^{(j)}(t)|^2 \,dt ,
	\end{equation}
	which holds for all absolutely continuous functions $h$ whose first $j-1$ derivatives are also absolutely continuous, and which obey $h(\tau)=h'(\tau)=\cdots=h^{(j-1)}(\tau)=0$ and hence for all $h\in W_0^j(0,\tau)$.
	It is then elementary to see that $|(|||f_n|||-|||f|||)|\le |||f_n-f|||\le \text{const}\times \|f_n-f\|_m \to 0$, completing the proof. 
\end{proof}

\section{The singularity theorem}
\label{sec:sing}

We now apply methods from \cite{Fewster:2019bjg} to prove a singularity theorem applicable to  $\fmax$-solutions ($\mathcal{M},g,\omega)$ to the SEE. We will assume that there is a Cauchy surface $S$ in $\mathcal{M}$ from which all normally emanating future-directed timelike geodesics of length $\tau$ are $T$-regular, with $T_t= T_0(1-t/\tau)$ for $T_0$ positive. Furthermore, we assume that the SEC holds for a short period of proper time along these geodesics after $S$ (this corresponds to `Scenario 1' in \cite{Fewster:2019bjg}). We will prove that if the initial extrinsic curvature scalar on $S$ is sufficiently negative then every such geodesic contains a focal point. It follows by standard arguments (summarised in~\cite{Fewster:2019bjg}) that the spacetime
is future timelike geodesically incomplete -- indeed, no future-directed timelike curve leaving $S$ can have length greater than $\tau$.  
In full detail, we will prove the following singularity theorem, which closely follows the structure of Theorems 4.2 and 4.4 of Ref.~\cite{Fewster:2019bjg}. 
\begin{theorem}
	\label{the:sing} 
	Let $(\mathcal{M},g,\omega)$ be a $\fmax$-solution of the SEE according to Def.~\ref{def:fsolSEE} 
	and let $S$ be a smooth spacelike Cauchy surface in $\mathcal{M}$. Suppose that:
	\begin{itemize}
		\item[(i)] there exist constants $\tau>0$ and $T_0>0$ so that every unit-speed timelike geodesic $\gamma:[0,\tau] \to \mathcal{M}$ emanating normally
		from $S$ is $T$-regular, with $T_t= T_0(1-t/\tau)$ (see Def.~\ref{def:Treg})
		\item[(ii)]
		there exist constants $\rho_0 \leq 0$ and $\tau_0 \in (0,\tau T_0/(T_0+2\tau)) $ so that, along each such $\gamma$, the inequality $R_{\mu \nu} \dot{\gamma}^\mu \dot{\gamma}^\nu|_{\gamma(t)} \leq \rho_0$ holds for $t \in [0,\tau_0]$;
		\item[(iii)] the extrinsic curvature $K$ of $S$ satisfies
		\be
		-K \geq \min\left\{\frac{n-1}{\tau_0}, \nu_*(T_0,\tau_0,\tau,\rho_0) \right\} \,,
		\ee
		with $\nu_*(T_0,\tau_0,\tau,\rho_0)$ given by Eq.~\eqref{eqn:Jest} below.
	\end{itemize}
	Then no future-directed timelike curve emanating from S has length greater than $\tau$ and $\mathcal{M}$ is future timelike geodesically incomplete.
\end{theorem}

The idea of the proof is to show that each unit-speed geodesic $\gamma:[0,\tau]\to\mathcal{M}$ emanating normally from $S$ contains a focal point. Under the hypotheses of the theorem, let $\rho(t)= R_{\mu\nu} U^\mu U^\nu|_{\gamma(t)}$ for some such $\gamma$,\footnote{By the SEE, $\rho(t)=-8\pi\langle\rho_U\rangle_\omega$, which is a slightly unfortunate collision of notation between the present paper and~\cite{Fewster:2019bjg}.} which defines a smooth function on $[0,\tau]$ because the state $\omega$ is Hadamard.
By Prop.~2.2 of Ref.~\cite{Fewster:2019bjg}, which brings together results from \cite{ONeill}, \cite{Fewster:2010gm} and references therein, $\gamma$ contains a focal point to $S$
if there is a  piecewise smooth $g:[0,\tau]\to\mathbb{R}$ with $g(0)=1$ and $g(\tau)=0$, such that
\begin{equation}
	J[g]\le -K|_{\gamma(0)}\,,
\end{equation}
where
\begin{equation}
	\label{eqn:Jdef}
	J[g]=\int_0^\tau\left( (n-1)\dot{g}(t)^2 + g(t)^2 \rho(t)  \right)\,dt\,.
\end{equation} 
We proceed to estimate $J[g]$ for the specific piecewise smooth function $g\in C^{m-1}([0,\tau])$ given by  
\begin{equation}\label{eq:gdef}
	g(t) = \begin{cases}
		1 & t\in [0,\tau_0) \\
		I(m,m;(\tau-t)/(\tau-\tau_0)) & t\in[\tau_0,\tau]\,,
	\end{cases}
\end{equation}
where $I(p,q;x)$ [commonly written $I_x(p,q)$] is the regularised incomplete Beta function.\footnote{
The function $g$ was also used in Sec.~4.1 of Ref.~\cite{Fewster:2019bjg}, where it was denoted $f$.}
Note that $g(0)=1$, $g(\tau)=0$. A plot of $g$ for $m=2$ is given in Fig.~\ref{fig:gphi}, which 
also shows the related function $\varphi:[0,\tau]\to\mathbb{R}$ defined by	
\begin{equation}\label{eq:scen1phidef}
	\varphi(t) = \begin{cases}
		I(m,m;t/\tau_0) & t\in [0,\tau_0) \\
		1 & t\in[\tau_0,\tau].
	\end{cases}
\end{equation}

\begin{figure}
	\centering
	\includegraphics[height=8cm]{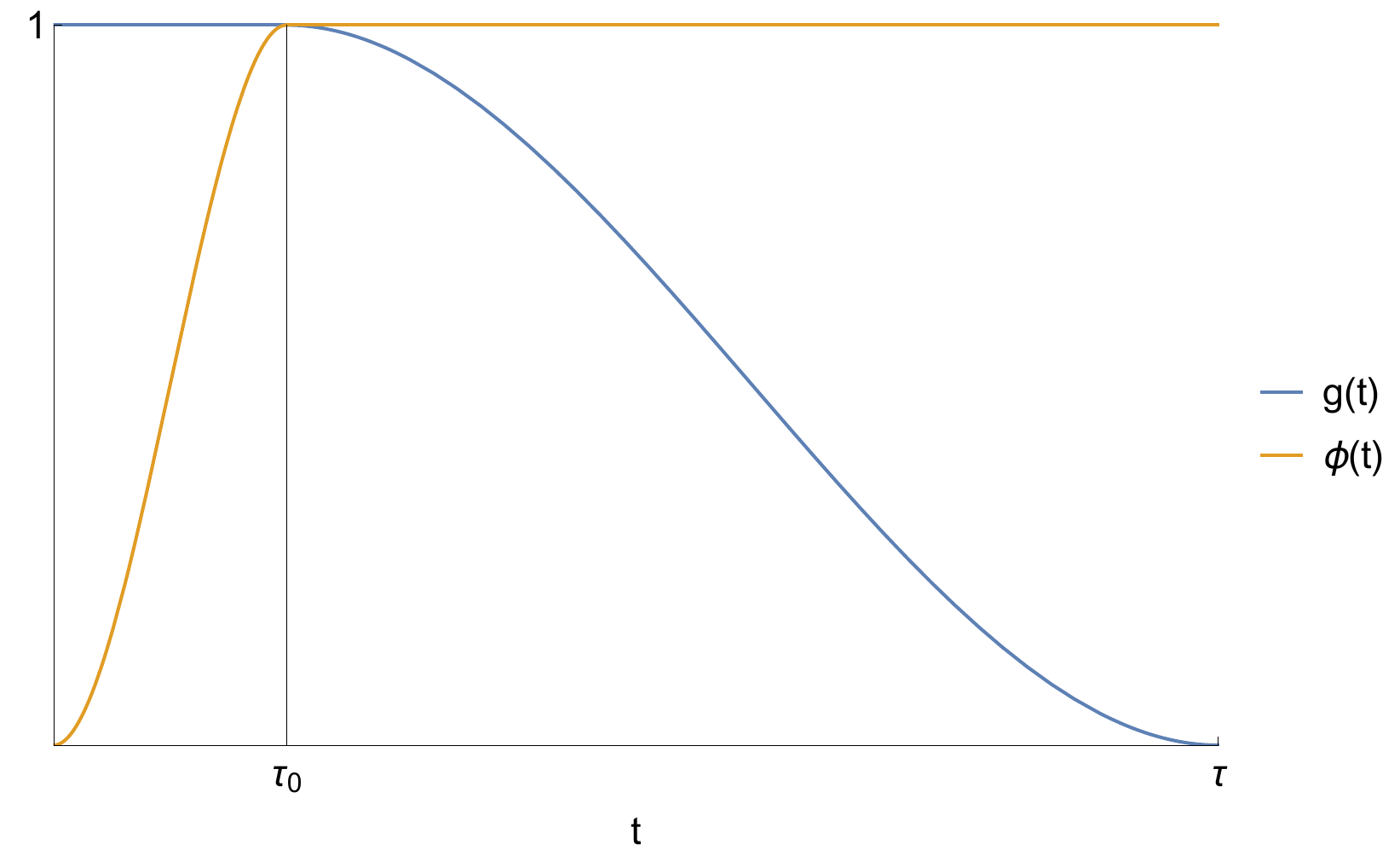}
	\caption{Plot of the functions $g(t)$ and  $\varphi(t)$ defined in Eq.~\eqref{eq:scen1fdef} and \eqref{eq:scen1phidef} respectively.}
	\label{fig:gphi}
\end{figure}

The advantage of these functions for our purposes is that their $L^2$-norms and those of their first derivatives are known in closed form~\cite{Fewster:2019bjg,connor2021integrals}. To be specific, 
the $L^2$-norms on the interval $[0,\tau]$ are 
\begin{equation}\label{eq:gnorms}
	\|\chi_{[\tau_0,\tau]}g\|^2 =(\tau-\tau_0)A_m, \qquad
	\|g'\|^2 = \frac{B_m}{\tau-\tau_0}, \qquad 
\end{equation}
(see the proof of Lem.~4.1 in~\cite{Fewster:2019bjg})
where $\chi_{[a,b]}$ is the characteristic function of $[a,b]$ and the constants $A_m$ and $B_m$ are 
\begin{equation}\label{eq:AB}
	A_m= \frac{1}{2} - \frac{(2m)!^4}{4(4m)!m!^4},\quad B_m=\frac{(2m-2)!^2(2m-1)!^2}{(4m-3)!(m-1)!^4}
\end{equation} 
for $m\in\mathbb{N}$. The first few relevant values are tabulated in Table~\ref{tab:const}. For the function $\varphi$ we have $\|\chi_{[0,\tau_0]}\varphi\|^2= \tau_0A_m$. It is then easily seen that for the product $f = g \varphi$ we have
\begin{equation}\label{eq:phigests}
	\|f\|^2 = A_m\tau \,.
\end{equation} 

\begin{table}
	\begin{center}\begin{tabular}{|c|c|c|c|c|} \hline
			$m$ & $1$ & $2$ & $3$ & $4$  \\ \hline \hline
			$A_m$ & $1/3$ & $13/35$ & $181/462$ & $521/1287$ \\
			$B_m$ & $1$ & $6/5$ & $10/7$ & $700/429$ \\
			\hline
		\end{tabular}\caption{The first few values of the constants $A_m$ and $B_m$.}
		\label{tab:const}
	\end{center}
\end{table}

As the function $f=g\varphi$ belongs to $W_0^m(0,\tau)$, and $\gamma$ satisfies the hypotheses of 
Theorem~\ref{the:Riccfin} due to assumption~(i) of Theorem~\ref{the:sing}, it may be shown that
\be
\label{eqn:rhohest}
\int_0^\tau \rho(t) f(t)^2 dt \leq Q_0 A_m \tau+Q_m H_m(T_0,\tau_0,\tau) \,,
\ee
where $H_m(T_0,\tau_0,\tau)=|||f|||^2$ is given by Eq.~\eqref{eqn:hestim} and the details of the derivation are discussed in Appendix~\ref{ssec:Hm}. Now it is straightforward to prove the following Lemma.

\begin{lemma}\label{lem:Jest} 
	Under the hypotheses of Theorem~\ref{the:sing}, if $\gamma:[0,\tau]\to\mathcal{M}$ is a unit-speed geodesic emanating normally from $S$ then the functional $J$ defined by~\eqref{eqn:Jdef} satisfies
	the estimate 
	\bea
	\label{eqn:Jest}
	J[g] \leq \nu_*(T_0,\tau_0,\tau,\rho_0) :=Q_m H_m(T_0,\tau_0,\tau)+Q_0 A_m \tau+(n-1) \frac{B_m}{\tau-\tau_0}+\rho_0 \tau_0 (1-A_m) \,.
	\eea 	
	for the function $g$ given by~\eqref{eq:gdef}. 
	Consequently, if $-K|_{\gamma(0)}\ge \nu_*$ then $\gamma$ contains a focal point to $S$.
\end{lemma}

	\begin{proof}
		Using $g^2=(\varphi g)^2+(1-\varphi^2)g^2=f^2+(1-\varphi^2)$ and Eq.~\eqref{eqn:rhohest} we have
		\bea
		\int_0^\tau dt g(t)^2 \rho(t)&=&\int_0^\tau dt \rho(t) f(t)^2+\int_0^\tau dt (1-\varphi^2)\rho(t) \nonumber \\
		&\leq& Q_m H_m(T_0,\tau_0,\tau)+Q_0A_m \tau+\int_0^{\tau_0} dt (1-\varphi^2) \rho(t)  \nonumber \\
		&\leq& Q_m H_m(T_0,\tau_0,\tau)+Q_0 A_m \tau+\rho_0 \tau_0(1-A_m) \,,
		\eea
		where we have also used the estimates for the Sobolev norms of $g$, $\varphi$ and $f$
		and assumption~(ii) of Theorem~\ref{the:sing}.
		Using 
	the estimate of $\|\dot{g}\|^2$ from Eq.~\eqref{eq:gnorms} the proof of~\eqref{eqn:Jest} is complete
	and the remaining statement follows by \cite[Prop.~2.2]{Fewster:2019bjg} as discussed above. 
	\end{proof}

The proof of our main result is now straightforward.
\begin{proof}[Proof of Theorem~\ref{the:sing}]
	Using assumptions~(i--iii) and Lemma~\ref{lem:Jest}, all unit-speed timelike geodesics of length $\tau$ emanating normally from $S$ contain a focal point to $S$. It follows that every unit-speed geodesic emanating normally from $S$ that does not contain focal points has length strictly less than $\tau$. Theorem~\ref{the:sing} now follows by standard reasoning, as summarised e.g., in Prop.~2.4(b) of~\cite{Fewster:2019bjg}.
\end{proof}

\section{A quantitative application using cosmological parameters}
\label{sec:cosmology}

To conclude, we will show that the theorems given above can be used quantitatively. We will use parameters drawn from the cosmology of our universe and show that the time-reversed version of our theorem predicts a past singularity, or singularity-like behaviour, on reasonable cosmological timescales. It should be borne in mind that this prediction relates to our semiclassical Einstein equation coupled to a single scalar quantum field. 	
Thus the aim is not to make predictions about our Universe but rather to show that our singularity theorem
can make physically plausible quantitative statements about a toy cosmology. This was also our approach
in~\cite{Brown:2018hym,Fewster:2019bjg} to which we will refer for some details.

The relevant information about our own universe is drawn from the $\Lambda$CDM cosmology model fitted with data from the PLANCK experiment~\cite{aghanim2018planck}. In particular, the ratios of dust and dark energy densities to the critical density are taken respectively as $\Omega_m=0.31$ and $\Omega_\Lambda=0.69$ at the present time. According to the model fitted with these parameters, the SEC was last satisfied when the Universe's age was approximately $2.41 \times 10^{17} \text{s}$ (at redshift parameter $z_*=0.642$), after which dark energy became dominant. The Hubble parameter at that time was $H_*=3.14 \times  10^{-18} \textrm{ s}^{-1}$ (for details on this calculation see \cite{Fewster:2019bjg}). Also at that time, the four-dimensional Ricci scalar had magnitude $|R_*|=5.7 \times 10^{-35} {\text s}^2$. What we draw from this is that (this model of) our own universe contains a Cauchy surface with extrinsic curvature $K=H_*$ towards the future, and that the SEC was obeyed for some period beforehand. 

Now let us examine the quantitative predictions of the time-reversed version of Theorem~\ref{the:sing}, in the four-dimensional case $m=2$. The idea is the Cauchy surface $S$ in the theorem should correspond to the $z_*=0.642$ hypersurface in the $\Lambda$CDM model. Several parameters are required: $M$ (the mass of the field, expressed as the inverse reduced Compton length), $\phi_{\text{max}}$ (the maximum magnitude of the scalar field), timescales $T_0$ and $\tau$ that determine the function $T_t$, together with parameters $\rho_0$ and $\tau_0$ relating to the period where SEC is assumed. Recall that our assumption is that the SEC should hold in the form $R_{\mu\nu}\dot{\gamma}^\mu\dot{\gamma}^\nu\le \rho_0\le 0$ for at least a period of time $\tau_0$ just before the surface $S$. We set $\rho_0=0$, so our assumption becomes simply the statement that SEC held for at least a period $\tau_0$. For $M$, we use values corresponding to three elementary particles: the pion, proton\footnote{The proton is not a scalar but QEIs for fermionic fields have been developed (e.g. \cite{Fewster:2001js,Dawson:2006py}) and they have a very similar form to the scalar field ones. The proton is of special interest in cosmology as the majority of baryonic matter is in the form of hydrogen.} and Higgs. In each case $\fmax$ is determined using the square root of the Wick square in a Minkowski spacetime KMS state of a temperature that is $1\%$ of the reduced Compton temperature $T_{\text{Compton}}=c\hbar M/k$ of the particle concerned, where $k$ is Boltzmann's constant. The temperature $T_{\text{Compton}}$ defines a scale beyond which the model cannot be trusted and, following~\cite{Brown:2018hym,Fewster:2019bjg} provides a heuristic basis for the value  
\begin{equation}
	\label{eqn:phimax}
	\fmax^2 \sim  10^{-2}\frac{c^4}{G_N} (M\ell_{\textrm{Pl}})^2  K_{1}(100) \,,
\end{equation}
where $K_\nu$ is a modified Bessel function of the second kind and $\ell_{\textrm{Pl}}$ is the Planck length.
We should note that Eq.~\eqref{eqn:phimax} is sensitive to the choice of temperature $10^{-2} T_{\text{Compton}}$ and that $\fmax$ would change significantly if this is raised or reduced. 
For the moment, we leave the timescales $T_0$, $\tau$ and $\tau_0$ as free parameters. Let us recall
that $T_t$ is the timescale over which the Minkowski QEI provides a valid estimate on the EED in the spacetime, and in particular that $T_0$ sets that timescale at the surface $S$. For Theorem~\ref{the:sing} we must assume $\tau_0< t_2 = \tau T_0/(T_0+2\tau)$.

The time-reversed version of Theorem~\ref{the:sing} asserts that, provided the past-directed unit-speed geodesics emanating normally from $S$ are $T$-regular, with $T_t=T_0(1-t/\tau)$, past geodesic incompleteness occurs within timescale $\tau$ if the extrinsic curvature of $S$ (to the future) exceeds $\min\{3/\tau_0, \allowbreak \nu_*(T_0,\tau_0,\tau,\rho_0) \}$, where 
$\nu_*$ was defined in Eq.~\eqref{eqn:Jest}. As $\rho_0=0$, 
\be
\label{eqn:nustarcosm}
\nu_*(T_0,\tau_0,\tau,0) = Q_2 H_2(T_0,\tau_0,\tau)+Q_0 A_2 \tau+  \frac{3B_2}{\tau-\tau_0} \,,
\ee
where $A_2= 13/35$, $B_2=6/5$ and, restoring dimensions, the constants $Q_2$ and $Q_0$ are given by
\be
Q_2=\hbar \frac{G_N }{2 \pi c^5} =4.60 \times 10^{-88} {\text s}^2 \,, \qquad Q_0= \frac{4\pi G_N M^2\phi_{\max}^2}{c^2} \,.
\ee
The values of $Q_0$ for the three elementary particles based on the $\fmax$ calculation described are given in Table~\ref{tab:particles}, while the function $H_2(T_0, \tau_0, \tau)$ is evaluated in Appendix~\ref{ssec:Hm} as  
\bea
H_2(T_0,\tau_0,\tau)&=& 2\left[  \sqrt{\frac{12}{\tau_0^3}}+\sqrt{\frac{12}{(\tau-\tau_0)^3}} +2H_1 \left(\sqrt{\frac{3}{10\tau_0t_2^2}}+c\sqrt{\frac{12}{(\tau-\tau_0)^3}} \right)\right.\nonumber \\
&&\left. +H_2\left( \sqrt{\frac{13\tau_0}{560t_2^4}}+c^2 \sqrt{\frac{13}{3(\tau-\tau_0)^3}} \right) \right]^2 \,,
\eea
where the numerical coefficients are $H_1=10.87$ and $H_2=234.6$
and $c = 2\tau^2/((2\tau-T_0)T_0)$

\begin{table}
	\begin{center}\begin{tabular}{|c|c|c|c|c|c|c|} \hline
			Particle & Mass in kg &  $T_{\text{Compton}}$ in $K$ & $Q_0$ in $s^{-2}$ & $\tau$ in $s$ & $\nu_*$ in $s^{-1}$ & min$T_0$ in $s$  \\ \hline \hline
			Pion & $2\times 10^{-28}$  & $1.56 \times 10^{12}$ & $2.38 \times 10^{-40}$  & $2.02 \times 10^{20}$ & $3.57 \times 10^{-20}$ & $1.05 \times 10^{-10}$ \\
			Proton & $1.67 \times 10^{-27}$  & $1.089 \times 10^{13}$ & $5.60 \times 10^{-37}$  & $4.16 \times 10^{18}$ & $1.73 \times 10^{-18}$ & $1.51 \times 10^{11}$ \\ 
			Higgs & $2.23 \times 10^{-25}$  & $1.45 \times 10^{15}$ & $1.78 \times 10^{-28}$& $2.33 \times 10^{14}$ & $3.09 \times 10^{-14}$ & $1.14 \times 10^{-13}$  \\ \hline
		\end{tabular}\caption{The temperature scale and $Q_0$ for three elementary particles with different masses, the pion, the Higgs and the proton.}
		\label{tab:particles}
	\end{center}
\end{table} 
 
The required extrinsic curvature depends on the three remaining independent parameters $T_0$, $\tau_0$ and $\tau$. We will consider situations in which $\tau$, the timescale within which the singularity occurs, is much greater than $T_0$, the timescale on which the Minkowski QEI is valid in the curved spacetime near $S$. This implies that $\tau_0<t_2\approx T_0/2\ll \tau$ and $c\approx \tau/T_0\gg 1$. 
If the extrinsic curvature at $S$ exceeds $3/\tau_0$, then Hawking's theorem already implies that past timelike geodesic incompleteness occurs within time $\tau_0$ of $S$, i.e., in the 
region where SEC holds. This rather trivial result is avoided by restricting to the situation where $\nu_*  \ll 3/\tau_0$. Our theorem now predicts that past timelike geodesic completeness fails within time $\tau$ provided $K>\nu_*$.  

For convenience, let us assume that the first term of Eq.~\eqref{eqn:nustarcosm} is much smaller than the other two (as actually happens for a broad range of parameters). Then the required extrinsic curvature is given by a combination of the last two terms,
\begin{equation}
	\nu_*(T_0,\tau_0,\tau,0) \approx  Q_0 A_2 \tau+  \frac{3B_2}{\tau}
\end{equation}
different for each particle, but roughly independent of $T_0$ and $\tau_0$ provided we remain in the regime mentioned above. 
The value of $\tau$ that minimizes $\nu_*$ and the corresponding value of $\nu_*$ are 
\begin{equation} \label{eq:taunustar}
\tau=\sqrt{3B_2/(Q_0 A_2)} = 3.11 Q_0^{-1/2}\,, \qquad \nu_*=\sqrt{12A_2B_2 Q_0} = 2.31Q_0^{1/2}\,.
\end{equation}  
Evaluating these expressions for the pion gives $\tau=2.02 \times 10^{20} {\text s}$ and $\nu_*=3.57 \times 10^{-20} {\text s}^{-1}$, while for the proton we have $\tau=4.16 \times 10^{18} {\text s}$ and $\nu_*=1.73 \times 10^{-18} {\text s}^{-1}$. Both values of $\nu_*$
are smaller than the value $H_*=3.14 \times  10^{-18} \textrm{ s}^{-1}$ drawn from $\Lambda$CDM cosmology. Therefore, in either of these models, we could conclude the presence of a past singularity within timescale $\tau$ if the extrinsic curvature took the value $H_*$ on $S$. Note that the timescales are larger, by 
$1$ order of magnitude (proton) or $3$ orders of magnitude (pion) than the age of our universe at $z_*=0.642$. Particularly for the proton, these results show that the theorem provides physically plausible quantitative estimates.  In the Higgs case, we have $\tau=2.33 \times 10^{14} {\text s}$ and $\nu_*=3.09 \times 10^{-14} {\text s}^{-1}$, very much larger than $H_*$. Similar behaviour was observed in an earlier hybrid model \cite{Fewster:2019bjg}.  

The parameters $\tau_0$ and $T_0$ are still free, though must be chosen to be much less than $\tau$ and to ensure that all the subsidiary assumptions made above still hold. Let us examine the condition that the first term in Eq.~\eqref{eqn:nustarcosm} be much smaller than the others, which sum to the approximate value of $\nu_*$ given in Eq.~\eqref{eq:taunustar}
in our situation of interest. Taking $\tau_0\approx T_0/2$ the assumption $T_0\ll \tau$ implies that 
$Q_2 H_2(T_0,\tau_0,\tau)\approx 4.77 \times 10^5 Q_2 \tau T_0^{-4}$. Using the above values, the first term  of \eqref{eqn:nustarcosm} is at least two orders of magnitude smaller than the sum of the other two provided $T_0\ge 1.31\times 10^{-20} Q_0^{-1/4}$, giving values of $1.06 \times 10^{-10}\text{s}$ for the pion, $1.51 \times 10^{-11}\text{s}$ for the proton and $1.14 \times 10^{-13}\text{s}$ for the Higgs. The largest of these values corresponds to a length scale of a few centimeters. By comparison, in the $\Lambda$CDM model at the $z_*=0.642$ surface, the four-dimensional Ricci scalar $|R_*|=5.7 \times 10^{-35} {\text s}^2$ gives a crude indication that curvature scales are on the order of $|R_*|^{-1/2}= 1.3 \times 10^{17}{\text s}$, while the Hubble parameter sets a timescale $H_*^{-1}= 3.2\times 10^{17}{\text s}$ of a similar order. As the minimum values of $T_0$ are many orders of magnitude below those scales it seems very reasonable to accept the validity of the Minkowski QEI. The $\tau$, $\nu_*$ and $T_0$ values for each particle are summarized in Table~\ref{tab:particles}. 

If, instead of taking $\tau_0\approx T_0/2$, we allow $\tau_0$ to be much less than $T_0$, then the contribution to $H_2$ equal to $24 Q_2/ \tau_0^3$ might become dominant. Requiring this term to remain at least two orders of magnitude smaller than $4.77 \times 10^5 Q_2 \tau T_0^{-4}$ we find a minimum value of $\tau_0$ for which our approximations remain valid. This reasoning gives a minimum value of $\tau_0$ for the pion $1.46 \times 10^{-21}\text{s}$, for the proton $4.00 \times 10^{-22}s$ and for the Higgs $1.53 \times 10^{-23} \text{s}$. In other words, the SEC need be assumed to hold only for a fleeting instant (this does not exclude it holding for a much longer period!). Furthermore, in all cases considered we have $3/\tau_0\gtrsim 6/T_0 \gg 6/\tau =0.84 \nu_*$ so we are well outside the regime (by several orders of magnitude) in which Hawking's result would predict geodesic incompleteness.

Finally, we note that there is still a lot of freedom that could be exploited. For example,
if we multiply the value of $\tau$ from that given in Eq.~\eqref{eq:taunustar} by a factor of $a$, the value of $\nu_*$ is multiplied by a factor of $a/2+1/(2a)$. In the case of the proton we may use this freedom to lower the singularity timescale for a proton to about $5.3$ times the age of our universe at the $z_*=0.642$ surface, while the required extrinsic curvature remains less than $H_*$.

\section{Conclusions}
\label{sec:conclusions}

In this work we presented the first proof of a semiclassical singularity theorem based directly on a QEI satisfied by a quantum field. In previous works we have derived, on the one hand, a singularity theorem with a weakened energy condition \cite{Fewster:2019bjg} and, on the other, a quantum strong energy inequality (QSEI) \cite{Fewster:2018pey}. Here we have connected those pieces, showing how a suitable energy condition for a singularity theorem can be derived from such a QEI. The main challenge we overcame was the need to approximate the curved spacetime QEI by its Minkowski counterpart, which necessitated a 
partition of unity argument into pieces of vanishingly short duration towards the end of
the geodesic segments under consideration, thus allowing for a possible blow-up in the curvature scales. 

In particular, we showed that a solution to the semiclassical Einstein equation with a massive scalar field cannot be future timelike complete if four conditions hold:
(i) an assumption essentially requiring that the local temperature of the scalar field remains below a specified threshold, (ii) all the unit speed timelike geodesics of
length $\tau$ emanating normally from a spacelike hypersurface $S$ are $T$-regular for a function $T_t=T_0(1-t/\tau)$ that decreases linearly to the future
(i.e., the Minkowski QEI provides a reliable estimate of the curved spacetime QEI on timescales given by $T_t$ at time $t$),
(iii)  the SEC holds for a short period just after $S$ and (iv) the extrinsic curvature of $S$ should obey a specific inequality. Failure of any of conditions (i)--(ii) might be regarded as indicating singular behaviour for many practical purposes, provided the temperature threshold is high enough -- for instance, failure of (ii) indicates that the 
curvature scales shrink to zero faster than linearly along at least one geodesic leaving $S$. 

To gain some insight into the potential utility of our theorem, it was applied using 
data drawn from $\Lambda$CDM cosmology and the PLANCK experiment~\cite{aghanim2018planck}, with the mass of the scalar field taken equal to 
that of the pion, the proton and the Higgs particle in turn. For the proton, we found physically reasonable parameter values so that an extrinsic curvature of a magnitude found at redshift $z=0.642$ in the $\Lambda$CDM model would guarantee a singularity within an order of magnitude of the actual age of the universe at that time. The SEC need only be assumed for a fleeting instant -- far shorter than would be needed for Hawking's result to predict a singularity. It should be emphasised that our result holds across \emph{all} solutions to the semiclassical Einstein equation satisfying the assumptions, so it is not at all surprising that the timescale to singularity is longer than in the physical cosmology. Nonetheless, it is gratifying that the simple model produces such physically plausible results.

A drawback in our result is the assumption of $T$-regularity, which works around
the fact that curved spacetime QEI in its current form requires the specification of 
a reference state. In fact, we believe it likely that our approach ends up overestimating the 
threshold on the extrinsic curvature needed to guarantee geodesic incompleteness. 
A more satisfactory approach would require a suitable `absolute' version of the QEI from~\cite{Fewster:2018pey} that does not require a reference state. Even then it would be necessary to estimate various terms in the Hadamard parametrix. Such a computation has been done in the case of the quantum weak energy inequality \cite{Kontou:2014tha} based on an absolute QEI from~\cite{FewsterSmith:2008}. The QSEI presents additional challenges as the bound is state dependent even in the minimally coupled scalar field case and is the topic of ongoing work. 

The singularity theorem presented here concerns timelike geodesic incompleteness. In Ref.~\cite{Fewster:2019bjg} we derived singularity theorems with weakened energy conditions for both timelike and null geodesics. However, QEIs along null geodesics in four-dimensional spacetime do not exist, as was shown by an explicit counterexample in~\cite{Fewster:2002ne}. The counterexample can be circumvented if a momentum cut-off is imposed, and it was recently shown in Ref.~\cite{Freivogel:2020hiz} that a conjectured null QEI with a UV momentum cutoff \cite{Freivogel:2018gxj} is sufficient to prove a null geodesic singularity theorem using methods from \cite{Fewster:2019bjg}. However, the QEI involved remains to be proven and any curvature corrections estimated. 

Finally, it is important to note that the semiclassical Einstein equation is a highly nontrivial system and relatively little is known about its solution space except where a high level of symmetry is imposed -- see e.g. \cite{Pinamonti:2013wya,Gottschalk:2018kqt,Sanders:2020osl}. It would be of interest to examine the behavior of solutions obtained in those works and the range of initial conditions that predict a primordial singularity.

\section*{Acknowledgements}

The authors would like to thank Rafael Sorkin for a useful conversation regarding the use of the semiclassical Einstein equation. This work is part of a project that has received funding from the European Union's Horizon 2020 research and innovation programme under the Marie Sk\l odowska-Curie grant agreement No. 744037 ``QuEST''. E-AK is supported by the ERC Consolidator Grant QUANTIVIOL.

	\appendix
\section{Details concerning the partition of unity}
\label{app:det}

In this Appendix we provide the details of the partition of unit construction summarized in Sec.~\ref{sec:part}.

\subsection{Construction of the partition of unity}\label{ssec:POUconstruction}

We aim to construct functions $\phi_n$ obeying
Eqs.~\eqref{eq:partition} and~\eqref{eq:intervalnests}, starting from 
a given function $t\mapsto T_t$. A first task is to show that a sequence $t_n$ of the required type can be found.
\begin{lemma}
	Suppose that $t\mapsto T_t$ is strictly positive, nonincreasing and continuous on $[0,\tau)$
	where $0<\tau\le\infty$, with $T_t\to 0$ as $t\to\tau$ if $\tau<\infty$. Then there is a strictly increasing sequence $t_n$  ($n\in\mathbb{N}_0$) in $[0,\tau)$ satisfying~\eqref{eq:intervalnests}, with $t_0=-\tfrac{1}{2}T_0$, $t_1=0$, $t_{n+1}-t_n=\tfrac{1}{2}T_{t_{n+1}}$ for $n\ge 0$ and $t_n\to \tau$ as $n\to\infty$. 	
	Furthermore, defining $n:[0,\tau)\to \mathbb{N}$ so that $t'\in [t_{n(t')},t_{n(t')+1})$ for all $t'\in [0,\tau)$, 
	the following estimates hold for all $t\in[0,\tau)$:
	\be\label{eq:tntest}
	t_{n(t)+1}\le  t + \tfrac{1}{2}T_t\,,\qquad T_{t_{n(t)+1}}\ge T_{t+T_t/2}\,.
	\ee
\end{lemma}
\begin{proof}
	We first show that there is a sequence $t_n\in [0,\tau)$ obeying $t_1=0$ and $t_{n+1}-t_n=\tfrac{1}{2}T_{t_{n+1}}$ for $n\ge 1$ (clearly this holds also for $n=0$ if we set $t_0=-\tfrac{1}{2}T_0$). If $\tau<\infty$ then set $T_t=0$ for $t\ge \tau$, thus extending $T_t$ as a continuous function. Suppose inductively that the first $n$ terms of the sequence can be constructed for $n\ge 1$ and all lie in $[0,\tau)$, noting that this holds for $n=1$. 
	The function $h(t)= t-\tfrac{1}{2}T_t$ is continuous on $[0,\infty)$, strictly increasing, obeys $h(t_n)<t_n$ and $h(t)\ge t- \tfrac{1}{2}T_{t_n}$ for $t\ge t_n$. Thus $h(\tau)>t_n$ and $h(t_n+ \tfrac{1}{2}T_{t_n})\ge t_n$. 
	By the intermediate value theorem, the equation $h(t)=t_n$ has a solution $t\in (t_n,t_n+ \tfrac{1}{2}T_{t_n}]\cap (0,\tau)$, which is in fact the unique solution in $(t_n,\infty)$.
	Therefore $t_n$ is well-defined as a strictly increasing sequence $t_n$ in $[0,\tau)$.   
	
	As $T_t$ is nonincreasing, the sequence $t_n$ also has the property $t_{n+1}-t_{n}= \tfrac{1}{2}T_{t_{n+1}}\le \tfrac{1}{2}T_{t_{n}} =t_{n}-t_{n-1}$. It follows that $[t_{n-1},t_{n+1}]\subset [t_n-T_{t_n}/2,t_n+T_{t_n}/2]$.
	
	If $\tau<\infty$, the sequence is increasing and bounded above, so $t_n\to t_*=\sup_n t_n\le \tau$, whereupon $T_{t_n}\to 0$ as $n\to\infty$ and hence $T_{t_*}=0$ by continuity, showing that $t_*\ge \tau$. Thus $t_n\to \tau$. Similarly, if $\tau=\infty$ but the sequence $t_n$ is bounded, then $t_n\to t_*<\infty$ and we again deduce $T_{t_*}=0$, contradicting the assumption that $T_t$ is strictly positive on $(0,\tau)$.
	
	To establish~\eqref{eq:tntest} one notes that $t_{n(t)}\le t<t_{n(t)+1}$ by definition, so  
	\be
	t_{n(t)+1}\le t +t_{n(t)+1}-t_{n(t)}= t +
	\tfrac{1}{2}T_{t_{n(t)+1}} \le t + \tfrac{1}{2}T_t\,,
	\ee
	as $T_t$ is non-increasing, from which we also deduce $T_{t_{n(t)+1}}\ge T_{t+T_t/2}$.
\end{proof}
We remark that one also has the very simple estimate $t_{n(t)+1}\le t+t_2$ for all $t$, on noting that $t_{n(t)+1}\le t+\tfrac{1}{2}T_{t_2}=t+t_2$
for $t\ge t_2$, and trivially, that $t_{n(t)+1}=t_2\le t+t_2$ for $t<t_2$.

The functions $\phi_n$ ($n\in\mathbb{N}$) corresponding to the sequence
$t_n$ are defined by Eq.~\eqref{eqn:phi}, where $F(x)=\sin(\theta(x))$ and $G(x)=\cos(\theta(x))$, with $\theta$ given by Eq.~\eqref{eqn:theta}. We must check that they are indeed smooth functions satisfying the conditions~\eqref{eq:partition} and~\eqref{eq:intervalnests}. 

Clearly, each $\phi_n$ takes values in $[0,1]$ with $\phi_n(t_n)=1$.
Additionally, $\phi_n$ is clearly smooth everywhere except perhaps $t=t_n$; but
as all derivatives of $\theta(x)$ vanish as $x\to 0$ or $x\to \tfrac{1}{2}$, it is easily seen that each $\phi_n$ is smooth at $t_n$ with all its derivatives vanishing there. 

As $\phi_n(t)$ vanishes identically for $t<t_n-T_{t_n}/2=t_{n-1}$
and for $t>t_n+T_{t_{n+1}}/2=t_{n+1}$, we have $\supp \phi_n\subset [t_{n-1},t_{n+1}]$.	Thus, on the interval $[t_n,t_{n+1}]$ only $\phi_n$ and $\phi_{n+1}$ can be nonzero, and one has
\[
\phi_n(t)^2 + \phi_{n+1}(t)^2 = \cos^2 [\theta((t-t_n)/T_{t_{n+1}})] + 
\sin^2[\theta((t-{t_{n+1}})/T_{t_{n+1}} +\tfrac{1}{2}) ] = 1
\]
because $t_{n+1}/T_{t_{n+1}}-1/2 = t_n/T_{t_{n+1}}$. Summarising, we have constructed
a sequence $\phi_n$ ($n\ge 1$) of smooth compactly supported functions taking values in $[0,1]$, with
$\supp \phi_n\subset [t_{n-1},t_{n+1}]$ and satisfying~\eqref{eq:partition} on
$\cup_{n=1}^\infty  [t_{n},t_{n+1}]=[0,\tau)$.

\subsection{Specific examples of the sequence $t_n$}\label{ssec:exampletn}

\paragraph{Finite interval} For a first example, define $T_t=T_0(1-t/\tau)$ on $[0,\tau)$ for $0<\tau<\infty$. The equation $t_{n+1}-t_n=\tfrac{1}{2}T_t$ can be rearranged to 
\begin{equation}
	t_{n+1}=\frac{t_n+T_0/2}{1+T_0/(2\tau)}\,,
\end{equation}
which is solved uniquely, subject to $t_1=0$, by 
\begin{equation}\label{eq:tnfinite}
	t_n= \tau \left(1 - \left(\frac{2\tau}{T_0+2\tau}\right)^{n-1}\right)
\end{equation}
for all $n\in \mathbb{N}_0$. In particular, $t_2=\tau T_0/(T_0+2\tau)$ and $t_0=-\tfrac{1}{2}T_0$.
Defining $n:[0,\tau)\to \mathbb{N}$ so that	$t'\in [t_{n(t')},t_{n(t')+1})$, one finds
\begin{equation}
	n(t') = 1 + \left\lfloor \frac{\log(1-t'/\tau)}{\log(2\tau/(T_0+2\tau))}\right\rfloor\,,
\end{equation}
from which $t_{n(t)+1}$ may be calculated. Here $\lfloor x\rfloor$ is the largest integer no greater than $x$. For convenience, one may use the estimate~\eqref{eq:tntest}  to give
\begin{equation}
	t_{n(t)+1} \le \left(1-\frac{T_0}{2\tau}\right) t + \frac{T_0}{2} \in [T_0/2, \tau)
\end{equation}	
and therefore
\begin{equation}\label{eq:Ttntfinite}
	T_{t_{n(t)+1}} \ge \frac{(\tau-t)(2\tau-T_0)T_0}{2\tau^2} = \frac{\tau-t}{c}, \qquad c = \frac{2\tau^2}{(2\tau-T_0)T_0}
\end{equation}
holds for all $t\in[0,\tau)$
For $0\le t<t_2$, one has $n(t)=1$ and therefore $T_{t_{n(t)}+1}=T_{t_2}=2t_2$ holds for $t\in [0,t_2)$.

\paragraph{Infinite interval}
As an example where $\tau=\infty$, consider $T_t=T_0 e^{-\lambda t}$ for some constant $\lambda>0$. Then the equation
\be
t_{n}-t_{n-1}=\tfrac{1}{2}T_{t_{n}}
\ee
may be written as a recurrence relation in terms of Lambert's $W$-function~\cite[\S4.13]{DLMF} by 
\be	\label{eq:exptn}
t_{n+1} = t_n + \lambda^{-1} W(\tfrac{1}{2}\lambda T_0e^{-\lambda t_n}),\qquad t_1=0.
\ee
In particular, $t_2 = W(\tfrac{1}{2}\lambda T_0)$ obeys $1< e^{\lambda t_2}< 1+\tfrac{1}{2}\lambda T_0$.

To obtain asymptotic approximations to $t_n$, it is convenient to write $t_n=s_n/\lambda$ where the sequence $s_n$ obeys
\be\label{eq:srec}
s_{n+1}-s_n= \mu e^{-s_{n+1}},\qquad s_1 = 0
\ee
and $\mu= \tfrac{1}{2}\lambda T_0$. Observe that the sequence $S_n=\log(\mu n+e^{-\mu})$ satisfies
\begin{equation}
	S_{n+1}-S_n = \mu\int_{n}^{n+1} \frac{dr}{\mu r + e^{-\mu}} > \frac{\mu}{\mu(n+1)+e^{-\mu}}=
	\mu e^{-S_{n+1}}, \qquad S_1 = 
	\log (\mu+ e^{-\mu})>0.
\end{equation}
Thus $S_1>s_1$ and, because 
\begin{equation}
	S_{n+1}- s_{n+1} > S_n-s_n +\mu \left(e^{-S_{n+1}}-e^{-s_{n+1}}\right),
\end{equation}
a simple induction argument shows that $S_{n}>s_n$ for all integer $n\ge 1$, i.e., one has
\begin{equation}
	s_n < \log (\mu n + e^{-\mu})\qquad n\ge 1.
\end{equation}
Now the relation~\eqref{eq:srec} implies
\begin{equation}
	s_n = \mu \sum_{r=2}^n e^{-s_r} 
\end{equation}
so the upper bound on $s_n$ gives
\begin{equation}
	s_n > \sum_{r=2}^n \frac{\mu}{\mu n+e^{-\mu}} = \psi(n+1+\mu^{-1}e^{-\mu}) - \psi(2+\mu^{-1}e^{-\mu}) =
	\log(n) - \psi(2+\mu^{-1}e^{-\mu}) +O(1/n)
\end{equation}
where $\psi(z)=\Gamma'(z)/\Gamma(z)$ is the digamma function.

Thus $s_n = \log(\mu n) +O(1)$ as $n\to\infty$ and the upper and lower bounds on $s_n$, together with~\eqref{eq:srec}, give
\begin{equation}
	\frac{\mu}{\mu n + e^{-\mu}}< s_{n+1}-s_n <
	\frac{\mu \exp\left(\psi(2+\mu^{-1}e^{-\mu})\right)}{\exp\left(\psi(n+2+\mu^{-1}e^{-\mu})  \right)} = \frac{\mu \exp\left(\psi(2+\mu^{-1}e^{-\mu})\right)}{\mu (n+1/2)+e^{-\mu}}\left(1+O(n^{-2})\right)
\end{equation}
i.e., $s_{n+1}-s_n=O(1/n)$. That is, our original sequence $t_n$ satisfies
\begin{equation}
	t_n = \lambda^{-1}\log(\tfrac{1}{2}\lambda T_0 n) + O(1), \qquad t_{n+1}-t_n = O(1/n)
\end{equation}
as $n\to\infty$.

\paragraph{Reverse engineering}
To conclude this section, we note that other example sequences can be obtained by reverse engineering the function $T_t$ from a required sequence $t_n$. 
Suppose a strictly increasing sequence $t_n$ ($n\in\mathbb{N}_0$) is given, with $t_1=0$
and $t_n\to\infty$ as $n\to\infty$. Choose any  
homeomorphism $k:[0,\infty)\to [0,2\tau/T_0+1)$ that obeys $k(0)=0$ and $k(1)=1$, and 
provides an interpolation of the sequence $t_n$ according to
\be\label{eq:tnfromk}
t_n=\tfrac{1}{2}(k(n)-1)T_0.
\ee 
Then it is easily seen, noting that
$k^{-1} (2t_n/T_0+1)=n$, that this sequence corresponds to 
\be
T_t =2t + T_0 - T_0 k\left(k^{-1}(1+2t/T_0)-1\right)
\ee
which, if $k$ is differentiable, obeys
$T'_0 = 2(1-k'(0)/k'(1))$.

For instance, the sequence~\eqref{eq:tnfinite} arises in this way via the homeomorphism $k:[0,\infty)\to [0,2\tau/T_0+1)$ given by 
\be
k(x)= 1 + \frac{2\tau}{T_0} \left(1 - \left(\frac{2\tau}{T_0+2\tau}\right)^{x-1}\right)\,,
\ee
while the sequence 
\be
\label{eqn:harmonictn}
t_n = \frac{T_0}{2}(\psi(n+1)+\gamma-1) = \begin{cases} \tfrac{1}{2}T_0 \left(-1+\sum_{r=1}^n r^{-1}\right) & n\ge 1 \\ -\tfrac{1}{2}T_0 & n=0
\end{cases},
\ee
where $\gamma$ is the Euler-Mascheroni constant,  
arises (via the homeomorphism of $[0,\infty)$ given by $k(x)=\psi(x+1)+\gamma$) from 
\be 
T_t =2t + T_0 \left(1-\gamma- \psi\left(\psi^{-1}(1-\gamma+2t/T_0)-1\right)\right).
\ee 

\subsection{Weighted Sobolev inequalities}\label{ssec:Sobolev}

Returning to a general partition of unity to start with, we can proceed to estimate the sum appearing on the LHS of Eq.~\eqref{eq:EEDpou}.  

Exchanging the integral with the sum we can write, for any $f\in C_0^\infty(0,\tau)$,
\begin{align}
	\label{eqn:sumfirst}
	\sum_{n=1}^\infty\|(f \phi_n)^{(m)}\|^2 &=\int_{0}^\infty \sum_{n=1}^\infty\left|\sum_{i=0}^m \binom{m}{i} f^{(m-i)} \phi_n^{(i)} \right|^2 dt \nonumber \\
	&  =\int_{0}^\infty \sum_{i=0}^m \sum_{j=0}^{m} \binom{m}{i}\binom{m}{j} \overline{h_i f^{(m-i)}}h_j f^{(m-j)} P_{ij}\, dt,\qquad P_{ij}(t)=\sum_{n=1}^\infty \frac{\phi_n^{(i)}(t) \phi_n^{(j)}(t)}{h_i(t) h_j(t)}  ,
\end{align}
where the nonnegative functions $h_j(t)$, to be chosen below, have dimensions of $[\text{time}]^{-j}$ to make the components $P_{ij}$ dimensionless. For each $t\ge 0$, there are at most two nonzero terms in the sum defining $P_{ij}$, namely the terms $r=n(t)$ and $r=n(t)+1$, where $n(t')$ is defined
so that $t'\in [t_{n(t')},t_{n(t')+1})$. (For example, if $t_n$ is related to an interpolating function $k$ as in Eq.~\eqref{eq:tnfromk}, then
\be
n(t) = \lfloor k^{-1}(2t/\tau_c+1)\rfloor\,,
\ee 
where $\lfloor x\rfloor$ denotes the greatest integer less than or equal to $x$.) Therefore the sum defining $P_{ij}$ is finite. Using the elementary inequality that $x^\dagger P x \le \|P\|_{\text{max}} (\sum_i |x_i|)^2$ where $\|P\|_{\text{max}}=\max_{i,j}|P_{ij}|$, the integral in \eqref{eqn:sumfirst} can be bounded as
\be
\label{eqn:sumn}
\sum_{n=1}^\infty\|(f \phi_n)^{(m)}\|^2 \leq \int_{0}^\tau \left(\sum_{i=0}^m  \binom{m}{i} h_i(t)  \left|f^{(m-i)}(t)\right|\right)^2  \|P(t)\|_{\text{max}} \,dt\,.
\ee
There is considerable freedom to choose the functions $h_i$. A convenient choice is to set 
\begin{equation}
	h_i(t) = \frac{H_i}{T_{t_{n(t)+1}}}, \qquad H_i = \max \{\| F^{(i)}\|_\infty, \|G^{(i)}\|_\infty\}.
\end{equation}
To see why, note that
\be
\phi_n^{(i)}(t)=\begin{cases}
	s_n^i F^{(i)}((t-t_n)/T_{t_n}+1/2)  & t< t_n \,,\\
	s_{n+1}^{i} G^{(i)}((t-t_n)/T_{t_{n+1}})  & t \geq t_n  \,,
\end{cases}
\ee
where $s_n=1/T_{t_n}$. At each $t\ge 0$, $\phi_{n(t)}^{(i)}(t)=s^i_{n(t)+1}G((t-t_n)/T_{t_{n+1}})$, while
$\phi_{n(t)+1}^{(i)}(t)=s^i_{n(t)+1}F((t-t_{n+1})/T_{t_{n+1}})$.
Thus, taking $h_i(t)= T_0^{-i} s^i_{n(t)+1} = T_{t_{n(t)+1}}^{-i}$, one has $|P_{ij}(t)| \le 2$ and hence the estimates
\be \label{eqn:sumn2}
\sum_{n=1}^\infty \|(f \phi_n)^{(m)}\|^2 \leq 2  \left\|\sum_{j=0}^m   \frac{H_j}{T_{t_{n(t)+1}}^{j}} \binom{m}{j} |f^{(m-j)}|  \right\|^2 
\leq 2  \left[\sum_{j=0}^m   H_j \binom{m}{j} \left\|\frac{f^{(m-j)}}{T_{t_{n(t)+1}}^{j}} \right\| \right]^2
\ee  
for all $f\in C_0^\infty(0,\tau)$, 
where we have used the triangle inequality in the last step.

To proceed, let us consider the specific partition of unity induced by
$T_t=T_0(1-t/\tau)$ on the interval $[0,\tau)$, now assuming $\tau<\infty$.
Suppose that $0<\tau_0<t_2$ and 
let $\chi$ be the characteristic function of $[0,\tau_0]$. 
Writing $f^{(m-j)}=\chi f^{(m-j)}+ (1-\chi)f^{(m-j)}$ and again
using the triangle inequality, 
\begin{equation}
	\left\|\frac{f^{(m-j)}}{T_{t_{n(t)+1}}^{j}} \right\| \le  \left\|\frac{\chi f^{(m-j)}}{T_{t_{n(t)+1}}^{j}} \right\| + \left\|\frac{(1-\chi)f^{(m-j)}}{T_{t_{n(t)+1}}^{j}} \right\| \,,
\end{equation}
also using the fact that $T_{t_{n(t)+1}}\ge 2t_2$ for $t\in [0,\tau_0]$ and otherwise using the estimate $T_{t_{n(t)+1}}\ge (\tau-t)/c$, (see~\eqref{eq:Ttntfinite} and~\eqref{eq:Ttntfinite}) one has
\begin{equation} 
	\label{eqn:sumgen}
	\sum_{n=1}^\infty \|(f \phi_n)^{(m)}\|^2 \leq 2   \left[\sum_{j=0}^m 
	H_j \binom{m}{j} \left( \frac{\|\chi f^{(m-j)}\|}{(2t_2)^j} + c^j\left\|\frac{(1-\chi)f^{(m-j)}}{(\tau-\cdot)^j}\right\|\right)
	\right]^2  
\end{equation}
for all $f\in C_0^\infty(0,\tau)$.  This is precisely the inequality~\eqref{eqn:fphiestim} needed to complete the proof of Theorem~\ref{the:Riccfin}. 

\subsection{A higher order Hardy inequality}\label{ssec:Hardy}

Let $U$ and $V$ be measurable real-valued functions on $[0,\infty)$ and suppose that $B>0$ defined by 
\[
B^2= \sup_{x>0}\left(\int_x^\infty |U(t)|^2\,dt\right)\left(\int_0^x \frac{dt}{|V(t)|^{2}}\right)
\]	
is finite. Where necessary, it is understood that the convention $0\cdot\infty=0$ is in force. Then there exists $C>0$ such that the weighted Hardy inequality 
\begin{equation}
	\int_0^\infty \left|U(x)\int_0^x f(t)dt \right|^2\,dx \le C^2 \int_0^\infty |V(x)f(x)|^2\,dx.
\end{equation}
holds for measurable, locally integrable $f$ and the least value of the constant $C$ for which this is true obeys $B\le C\le 2B$. In particular, the finiteness of the right-hand side implies that the left-hand side is finite. 
See \cite{Muckenhoupt:1972} for an elementary proof. Considerably more general results are known.

We apply this inequality in the case $U(x)=x^{-n}H(1-x)$, $V(x)= x^{-(n-1)}H(1-x)$ where $H$ is the Heaviside function, for which
\[
B^2 = \sup_{x>0}\frac{H(1-x)(x^{-(2n-1)}-1)}{(2n-1)}\frac{x^{2n-1}}{(2n-1)}=\frac{1}{(2n-1)^2}
\]
and therefore
\begin{equation}\label{eq:ex2}
	\int_0^1 \left|\frac{1}{x^n}\int_0^x f(t)dt \right|^2\,dx \le \frac{4}{(2n-1)^2} \int_0^1 \frac{|f(x)|^2}{x^{2(n-1)}}\,dx.
\end{equation}

Inequality~\eqref{eq:ex2} can be rephrased as the statement
\begin{equation}
\left\| x\mapsto \frac{h(x)}{x^n}\right\|^2 \le \frac{4}{(2n-1)^2} \left\|x\mapsto \frac{h'(x)}{x^{2(n-1)}}\right\| ^2
\end{equation}
for all $h\in AC([0,1])$ obeying $h(0)=0$. Iterating, one finds
\begin{equation} 
	\left\| x\mapsto \frac{h(x)}{x^n}\right\|^2 \le \frac{4^n}{((2n-1)!!)^2}   \|h^{(n)}\|^2 
\end{equation}
for all $h\in AC([0,1])$ such that each $h^{(j)}\in AC([0,1])$ with $h^{(j)}(0)=0$ for
$0\le j\le n-1$. Inequality~\eqref{eq:higherHardy} is now immediate by an elementary change of integration variable.

See e.g.,~\cite{KufnerKulievPersson:2012} for more general (though less explicit) higher order Hardy inequalities.

\subsection{Computation of $H_m(T_0,\tau_0,\tau)$}\label{ssec:Hm}

In Section~\ref{sec:sing} we need to compute $H_m(T_0,\tau_0,\tau)=|||f|||^2$, for the function $f$ built from incomplete beta functions
\begin{equation}\label{eq:scen1fdef}
	f(t) = \begin{cases}
		I(m,m;t/\tau_0) & t\in [0,\tau_0) \\
		I(m,m;(\tau-t)/(\tau-\tau_0)) & t\in[\tau_0,\tau]\,.
	\end{cases}
\end{equation} 
 Elementary manipulations give 
\be
\|\chi f^{(m-j)}\|^2 = \frac{A_{m,m-j}}{\tau_0^{2(m-j)-1}}\,, \qquad
A_{m,k}:= \int_0^1 I^{(k)}(m,m;x)^2\,dx
\ee
and
\be
\left\|\frac{(1-\chi)f^{(m-j)}}{(\tau-\cdot)^j}\right\|^2 = \frac{C_{m,j}}{(\tau-\tau_0)^{2m-1}}\,,\qquad C_{m,j} = \int_0^1 x^{-2j}I^{(m-j)}(m,m;x)^2\,dx
\ee
Here, conventions have been chosen so that $A_{m,0}=A_m$ and $C_{m,0}=C_m$, in the notation used in \cite{Fewster:2019bjg} and Section~\ref{sec:sing}.
Overall, this gives
\begin{equation}
	\label{eqn:hestim}
	H_m(T_0,\tau_0,\tau)
	:=2\left[\sum_{j=0}^m 
	H_j \binom{m}{j} \left( \sqrt{\frac{A_{m,m-j}}{\tau_0^{2(m-j)-1}(2t_2)^{2j}}} + c^j
	\sqrt{\frac{C_{m,j}}{(\tau-\tau_0)^{2m-1}}}\right)
	\right]^2 
\end{equation} 
for the function $f$ in~\eqref{eq:scen1fdef}. Here, the dependence on $T_0$ enters via $c$ and $t_2$.

\begin{table}
	\begin{center}\begin{tabular}{|c|c|c|c|} \hline
			$j$ & $0$ & $1$ & $2$   \\ \hline \hline
			$A_{2,j}$ & $13/35$ & $6/5$ & $12$   \\
			$C_{2,j}$ & $12$ & $12$ & $13/3$  \\ \hline
		\end{tabular}\caption{The values of the constants $A_{m,j}$ and $C_{m,j}$ for $m=2$.}
		\label{tab:moreconst}
	\end{center}
\end{table}

To make this bound quantitative, one needs to evaluate the constants $A_{m,j}$ and $C_{m,j}$.
Our main interest will be in the case $m=2$, for which the relevant values are tabulated in Table~\ref{tab:moreconst}. Using these values together with Eq.~\eqref{eqn:hestim} gives
\bea
H_2(T_0,\tau_0,\tau)&=& 2\left[ H_0\left( \sqrt{\frac{12}{\tau_0^3}}+\sqrt{\frac{12}{(\tau-\tau_0)^3}}\right)+2H_1 \left(\sqrt{\frac{3}{10\tau_0t_2^2}}+c\sqrt{\frac{12}{(\tau-\tau_0)^3}} \right)\right.\nonumber \\
&&\left. +H_2\left( \sqrt{\frac{13\tau_0}{560t_2^4}}+c^2 \sqrt{\frac{13}{3(\tau-\tau_0)^3}} \right) \right]^2 \,.
\eea
From Eq.~\eqref{eqn:Hmax} we have $H_0=1$ and 
\be
H_1=\sup{\theta'(x)} \,, \qquad H_2=\sup{\theta'(x)^2}+\sup{\theta''(x)} \,.
\ee
Now $\theta'(x)=e^{-1/x} e^{-1/(1/2-x)}$ on $[0,1/2]$ and vanishes otherwise. Differentiating,
$\theta''(x)=\theta'(x) (1-4x)/(1-2x)^2$, so it is easily seen that the maximum of $\theta'$ occurs at $x=1/4$, giving $H_1=10.87$. Differentiating further, 
\begin{equation}
	\theta'''(x) = \theta'(x){\frac {48\,{x}^{4}-48\,{x}^{3}+32\,{x}^{2}-10\,x+1}{{x}^{4} \left( -1
			+2\,x \right) ^{4}}}
\end{equation} 
and one may determine that the maximum of $\theta''$ occurs at $x= {\tfrac{1}{4}}-{\tfrac {\sqrt {6\,\sqrt {13}-21}}{12}}$ with value $116.5$ (see  Fig.~\ref{fig:theta}). Consequently $H_2=234.6$.   
\begin{figure}
	\centering
	\includegraphics[height=4.5cm]{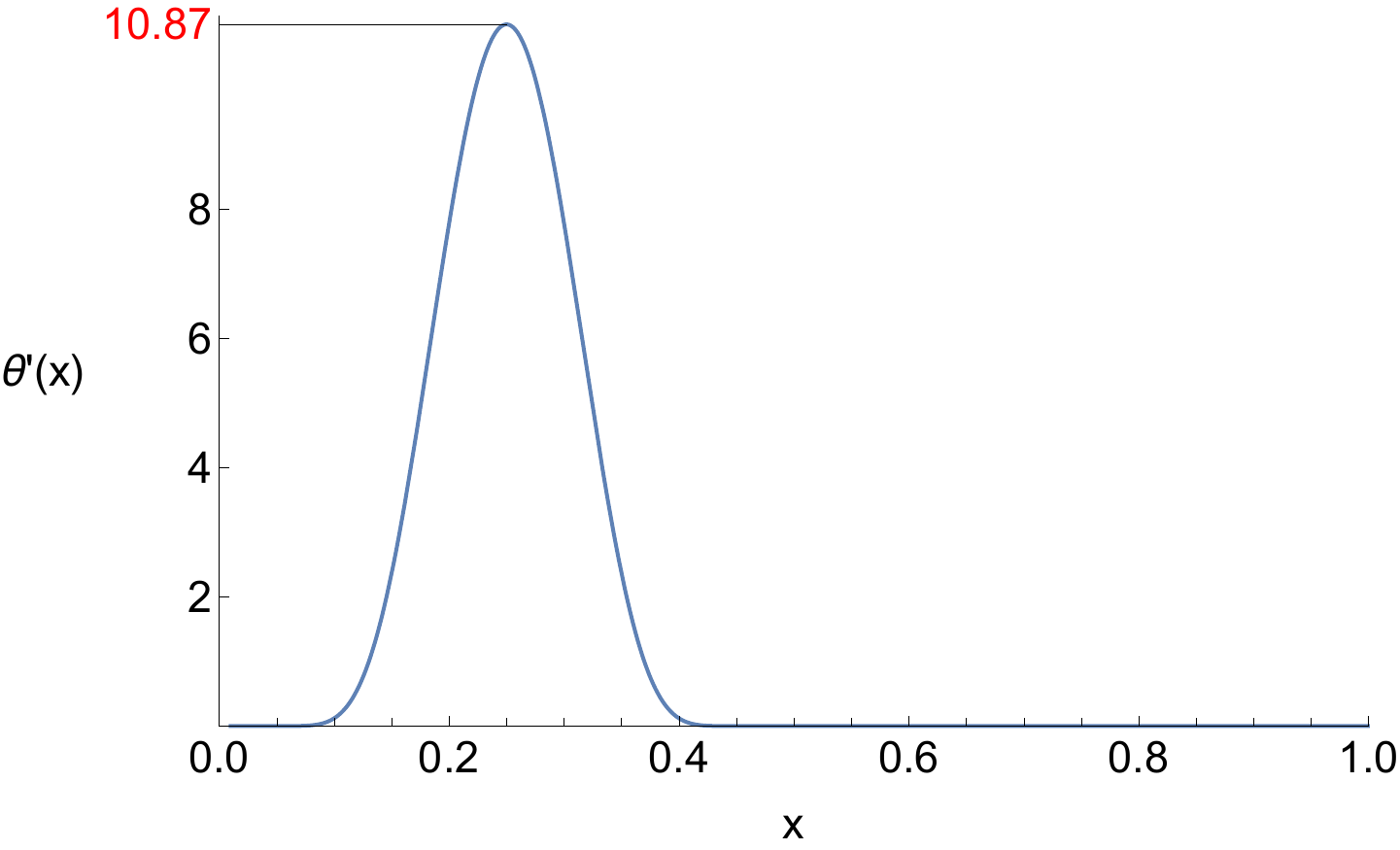}
	\includegraphics[height=4.5cm]{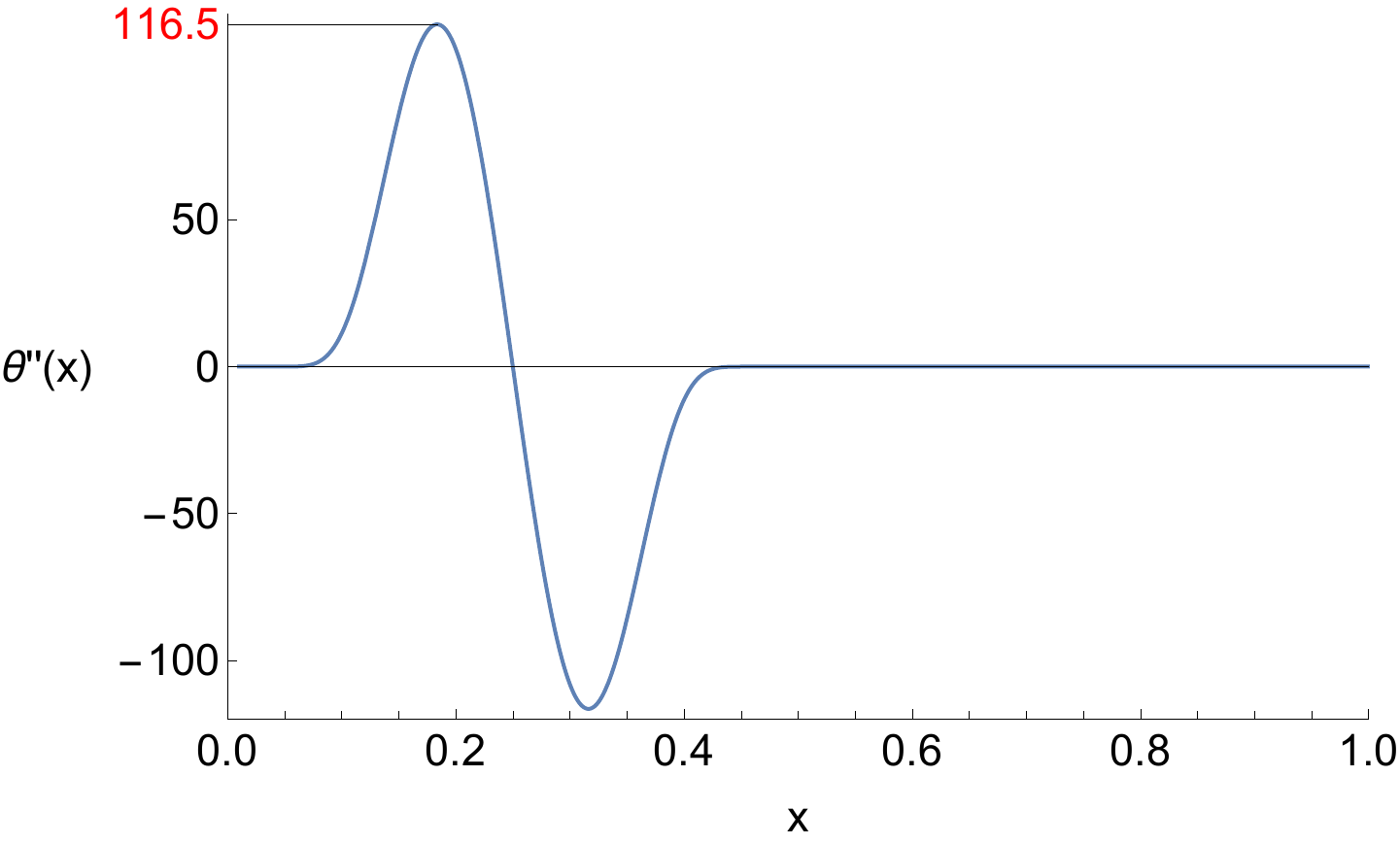}
	\caption{Plots of the first (left) and second (right) derivatives of function $\theta$ (see Eq.~\eqref{eqn:theta}). Their maximum values are indicated in red.}
	\label{fig:theta}
\end{figure}

More generally, closed forms for $A_{m,0}=A_0$ and $C_{m,0}=C_0$ are known for all $m\in\mathbb{N}$. 
For general $k\ge 1$, we have
\begin{align}
	A_{m,k}&= \frac{1}{B(m,m)^2}\int_0^1 \left(\frac{d^{k-1}}{dx^{k-1}} x^{m-1}(1-x)^{m-1}\right)^2\,dx \nonumber\\
	&=\frac{\Gamma(2m)^2}{\Gamma(m)^2(m-k)!^2}\int_0^1 \left(\sum_{r=0}^{m-1} 
	\frac{(1-m)_r \, (m)_r}{(m-k+1)_r}\frac{x^{m-k+r}}{r!}
	\right)^2\,dx\nonumber\\
	&=\frac{\Gamma(2m)^2}{\Gamma(m)^2(m-k)!^2(2(m-k)+1)}
	\sum_{r=0}^{\infty}\sum_{s=0}^{\infty}\frac{(1-m)_r \, (m)_r}{(m-k+1)_r}\frac{(1-m)_s \, (m)_s}{(m-k+1)_s}\frac{(2(m-k)+1)_{r+s}}{(2(m-k+1))_{r+s}}\frac{1}{r!s!}\nonumber\\
	&=\frac{\Gamma(2m)^2}{\Gamma(m)^2 (m-k)!^2 (2(m-k)+1)}
	\KdF{1:2;2}{1:1;1}{2(m-k)+1}{m,1-m}{m,1-m}{2(m-k+1)}{m-k+1}{m-k+1}{1,1},
\end{align}
(note that the series actually terminate after $m$ terms) where the general Kamp\'e de F\'eriet function~\cite{KampedeFeriet1921a,KampedeFeriet1921b} in $2$ variables is given by
\begin{equation}
	\label{eq:KdFdefn} 
	\KdF{A:C;F}{B:D;G}{(a)}{(c)}{(f)}{(b)}{(d)}{(g)}{x,y}  = 
	\sum_{m=0}^\infty\sum_{n=0}^\infty\frac{((a))_{m+n} ((c))_m ((f))_n}{((b))_{m+n}((d))_m ((g))_n }\, \frac{x^m y^n}{m!n!}
	\,,
\end{equation}
in which $(a)$, $(c)$ etc are sequences of length $A$, $C$ etc, and the notation $((c))_m=\prod_{i=1}^C (c_i)_m$ denotes the product of Pochhammer symbols. A similar calculation yields
\begin{equation}
	C_{m,j}= \frac{\Gamma(2m)^2}{j!^2\Gamma(m)^2}
	\KdF{1:2;2}{1:1;1}{1}{m,1-m}{m,1-m}{2}{j+1}{j+1}{1,1}.
\end{equation}
In certain cases, these expressions can be evaluated in terms of more familiar special functions. For example, direct calculation gives
\begin{align}
	\int_0^1 x^{\lambda-1}I'(\mu,\mu';x)^2\,dx &=\frac{1}{B(\mu,\mu')^2}\int_0^1 x^{\lambda-1} x^{2\mu-2}(1-x)^{2\mu'-2}\,dx \\
	&=\frac{B(\lambda+2\mu-2,2\mu'-1)}{B(\mu,\mu')^2} 
\end{align}	
so with $\mu'=\mu$ and $\lambda=3-2\mu$,
\begin{equation}
	\int_0^1 x^{-2(\mu-1)}  I'(\mu,\mu';x)^2\,dx =\frac{B(1,2\mu-1)}{B(\mu,\mu)^2} = \frac{1}{(2\mu-1)B(\mu,\mu)^2}.
\end{equation}
giving $C_{m,m-1}$ in the case $\mu=m$.  

The coefficients $C_{m,m}$ can also be obtained. In the notation of~\cite{connor2021integrals}, one has
	\begin{equation}
		\int_0^1 x^{\lambda-1}I(\mu,\mu',x)^2 dx = B(\lambda,1)\mathcal{B}(\lambda,1,\mu,\mu',\mu,\mu')
	\end{equation}
	and one may read off from that reference a formula for the right-hand side 
	in terms of a generalized hypergeometric function:
	\begin{equation}
		\int_0^1 x^{\lambda-1}I(\mu,\mu',x)^2 dx = \frac{1}{\lambda}\left(1-\frac{2B(\lambda+2\mu,\mu')}{\mu B(\mu,\mu')^2} \, \pFq{3}{2}{1-\mu',\mu,\lambda+2\mu}{1+\mu,\lambda+2\mu+\mu'}{1}\right),
	\end{equation}
	which was proved for nonnegative parameters $\lambda,\mu\,\mu'$.
	Setting $\mu'=\mu$ and analytically continuing the above formula to $\lambda=1-2\mu$ one finds
	\begin{equation}
		\int_0^1 x^{-2\mu}I(\mu,\mu,x)^2 dx = \frac{1}{2\mu-1}\left(\frac{2}{\mu^2 B(\mu,\mu)^2}  \pFq{3}{2}{1,\mu,1-\mu}{1+\mu,1+\mu}{1}-1\right)\,,
	\end{equation}
	yielding a closed form for $C_{m,m}$ on setting $\mu=m$.

\bibliographystyle{utphys}
\bibliography{sing}

\end{document}